\RequirePackage{etoolbox}
\newtoggle{conference}

\togglefalse{conference} % For the arxiv version

\newcommand{\inConference}[1]{\iftoggle{conference}{#1}{}} % For text that should appear in the conference version only
\newcommand{\inArxiv}[1]{\iftoggle{conference}{}{#1}}  % For text that should appear in the ArXiv version only

\inArxiv{\documentclass[11pt]{article}}
\inConference{
	\documentclass[twoside,leqno,twocolumn]{article}  
	\usepackage{ltexpprt}

	\newcommand{\qedhere}{}
}

\usepackage{graphicx}
\usepackage{latexsym}
\usepackage{amssymb}
\usepackage{amsmath}
\inArxiv{
	\usepackage{amsthm}
	\usepackage[margin=1in]{geometry}
	\usepackage[ruled,vlined,linesnumbered]{algorithm2e}
}
\inConference{
	\usepackage[ruled,vlined,linesnumbered,norelsize]{algorithm2e}
}
\usepackage{paralist}
\usepackage{nicefrac}

\begin{document}

\newcommand {\ignore} [1] {}

\inArxiv{
\newtheorem{theorem}{Theorem}[section]
\newtheorem{lemma}[theorem]{Lemma}
\newtheorem{fact}[theorem]{Fact}
\newtheorem{corollary}[theorem]{Corollary}
\newtheorem{definition}{Definition}[section]
\newtheorem{proposition}[theorem]{Proposition}
\newtheorem{observation}[theorem]{Observation}
\newtheorem{claim}[theorem]{Claim}
\newtheorem{assumption}[theorem]{Assumption}
\newtheorem{notation}[theorem]{Notation}
\newtheorem{reduction}{Reduction}
}
\inConference{
	\newtheorem{observation}{Observation}
}

\def \aa   {\alpha}
\def \bb   {\beta}
\def \gg   {\gamma}
\def \ee   {\varepsilon}
\def \el   {\ell}
\def \ss   {\sigma}
\def \dd   {\delta}
\def \Om   {\Omega}

\def \PP   {{\cal P}}
\def \QQ   {{\cal Q}}
\def \DD   {{\cal D}}
\def \NN   {{\cal N}}
\def \AA   {{\cal A}}
\def \MM   {{\cal M}}
\def \II   {{\cal I}}
\def \TT   {{\cal T}}

\newcommand{\ie}{{\it i.e.}}
\newcommand{\eg}{{\it e.g.}}
\newcommand{\MC}{{\texttt{Max Cut}}}
\newcommand{\MDC}{\texttt{Max DiCut}}
\newcommand{\MkC}{\texttt{Max $k$-Coverage}}
\newcommand{\GA}{\texttt{Generalized Assignment}}
\newcommand{\FL}{\texttt{Max Facility Location}}
\newcommand{\MB}{\texttt{Max Bisection}}
\newcommand{\RSet}{{\mathtt{R}}}
\newcommand{\opt}{{\mathtt{opt}}}

\pagenumbering{arabic}

\title{Comparing Apples and Oranges:\\Query Tradeoff in Submodular Maximization}

\author{
 Niv Buchbinder\thanks{Statistics and Operations Research Dept., Tel Aviv University, Israel. E-mail: \texttt{niv.buchbinder@gmail.com}. Research supported by ISF grant 954/11 and BSF grant 2010426.}
 \and
 Moran Feldman\thanks{School of Computer and Communication Sciences, EPFL, Switzerland. E-mail: \texttt{moran.feldman@epfl.ch}. Research supported in part by ERC Starting Grant 335288-OptApprox.}
 \and
 Roy Schwartz\thanks{Dept. of Computer Science, Princeton University, Princeton, NJ. E-mail: \texttt{roysch@cs.princeton.edu}.}
}

\inConference{\date{}}
\maketitle

\begin{abstract}
\expandafter\csname\inConference{small}\endcsname
\inConference{\baselineskip=9pt}
Fast algorithms for submodular maximization problems have a vast potential use in applicative settings, such as machine learning, social networks, and economics.
Though fast algorithms were known for some special cases, only recently Badanidiyuru and Vondr\'{a}k \cite{BV14} were the first to explicitly look for such algorithms in the general case of maximizing a monotone submodular function subject to a matroid independence constraint.
The algorithm of Badanidiyuru and Vondr\'{a}k matches the best possible approximation guarantee, while trying to reduce the number of value oracle queries the algorithm performs.

Our main result is a new algorithm for this general case which establishes a surprising {\em tradeoff} between two seemingly unrelated quantities: the number of value oracle queries and the number of matroid independence queries performed by the algorithm.
Specifically, one can decrease the former by increasing the latter and vice versa, while maintaining the best possible approximation guarantee.
Such a tradeoff is very useful since various applications might incur significantly different costs in querying the value and matroid independence oracles.
Furthermore, in case the rank of the matroid is $O(n^c)$, where $n$ is the size of the ground set and $c$ is an absolute constant smaller than $1$,
the total number of oracle queries our algorithm uses can be made to have a smaller magnitude compared to that needed by~\cite{BV14}.
%our algorithm can use a smaller magnitude of oracle queries compared to that needed by~\cite{BV14}.
We also provide even faster algorithms for the well studied special cases of a cardinality constraint and a partition matroid independence constraint,
both of which capture many real-world applications and have been widely studied both theorically and in practice.
%improving the state of the art for these two widely used constraints.

%Our main result is a somewhat more involved fast algorithm for this problem having the following interesting property: value oracle queries can be traded for independence oracle queries. This property is very useful when independence and value oracle queries have different costs. Moreover, when the rank of the matroid is not very large (compared to the size of the ground set), our algorithm can use strictly less oracle queries in total. We also describe fast algorithms for cardinality and partition matroids improving on the state of the art results for these constraints in terms of the number of oracle queries.
\end{abstract}

\inArxiv{
	\thispagestyle{empty}
	\newpage
	\setcounter{page}{1}
}

\section{Introduction}

The study of combinatorial optimization problems with a submodular objective has attracted much attention in recent years, as submodular functions arise naturally in various disciplines, \eg, combinatorics, economics, and machine learning.
Many well-known problems in combinatorial optimization are in fact submodular maximization problems, including: {\MC} \cite{GW95,H01,K72,KKMO07,TSSW00}, {\MDC} \cite{FG95,GW95,HZ01}, {\GA} \cite{CK05,CKR06,FV06,FGMS06}, {\MkC} \cite{F98,KMN99}, {\MB} \cite{ABG13,FJ95}, and {\FL} \cite{AS99,CFN77a,CFN77b}.
Furthermore, practical applications of submodular maximization problems are common in social networks \cite{HMS08,KKT03}, vision \cite{BJ01,JB11}, machine learning \cite{KSG08,KG05,KLGVF08,LB10,LB11} (the reader is referred to a comprehensive survey by Bach \cite{Bach13}), and many other areas.
Elegant algorithmic techniques were developed in the course of this line of research which achieved provable, and in some cases even tight, approximation guarantees.
A prime example for the latter is the {\em continuous greedy} algorithm of~\cite{CCPV11} for maximizing a monotone submodular function subject to a matroid independence constraint.
Unfortunately, most of these techniques result in algorithms which are efficient in theory but are not practical.
Hence, a natural research question is whether one can obtain {\em faster} algorithms with provable tight guarantees for basic submodular optimization problems.

How does one measure the speed of an algorithm for a submodular maximization problem?
%How does one measure how fast is an algorithm for maximizing a submodular function $f$?
Since an explicit representation of the submodular function might be exponential in the size of its ground set, the algorithm is assumed to access the objective function $f$ via a {\em value oracle}\footnote{Other types of oracles exist, however, value oracles are the most commonly used type in the literature.} which returns the value of $f(S)$ given any subset $S$ of the ground set. Usually, the number of value oracle queries dominates the number of arithmetic operations in the algorithm up to a polylogarithmic factor. Hence, it is natural to use the number of value oracle queries as a measure for the algorithm's speed. The use of this measure is also facilitated by the observation that implementions of the value oracle have a non-neglagible time complexity in many applications.

Badanidiyuru and Vondr\'{a}k \cite{BV14} were the first to consider the question of finding fast algorithms with provable guarantees for maximizing a submodular function in its full generality.
They presented algorithms that achieve an almost tight approximation guarantee of $1-\nicefrac[]{1}{e}-\varepsilon$, for any $\varepsilon > 0$, for both the cardinality and the more general matroid independence constraints. The algorithms they designed use $O\left( \frac{n}{\varepsilon}\log{\left( \frac{n}{\varepsilon}\right)}\right)$ value oracle queries for the cardinality constraint and $O\left( \frac{nk}{\varepsilon ^4}\log ^{2}{\left( \frac{n}{\varepsilon}\right)}\right)$ value oracle queries for a general matroid independence constraint.
Here $k$ denotes the rank of the matroid and $n$ is the size of the ground set.

In the context of a simple constraint such as a cardinality constraint, it is easy to determine whether a given solution $S$ is feasible.
However, when considering more complex constraints, such as a general matroid independence constraint, one usually assumes the existence of an {\em independence oracle}.
This oracle determines whether a given subset $S$ of elements of the ground set is independent in the matroid, \ie, feasible.
In all previous works, as far as we know, the number of value oracle queries dominates the number of independence oracle queries, and thus, the latter is usually overlooked.
This overlook is unfortunate since the implementation of these two {\em distinct} oracles in various applications might have running times of completely different magnitudes.\footnote{For example, the independence oracle can be implemented very efficiently when the constraint is a uniform or partition matroid. However, no linear time implementation is known when the constraint is a matching or linear matroid.}
In particular, minimizing the number of value oracle queries, as implicitly done by previous works, might not be the correct goal.
Furthermore, it is not clear whether the two different goals mentioned above, \ie, minimizing the number of value oracle queries and minimizing the number of independence oracle queries, are related.

\subsection{Our Results}

%\subsubsection{Main Result}

Our main result is the design of an algorithm for maximizing a monotone submodular function subject to a matroid independence constraint.
Our algorithm establishes a {\em tradeoff} between the following two seemingly unrelated quantities: the number of value oracle queries and the number of independence oracle queries performed by the algorithm. The following theorem summarizes the result.

\begin{theorem} \label{th:general_matroid}
There exists an algorithm that given a non-negative monotone submodular function $f : 2^\NN \to \mathbb{R}^+$, a matroid $\MM = (\NN, \II)$ of rank $k$, and parameters $\ee > 0$ and $\lambda \in [1, k]$, finds a solution $S\in \II$ where:
%There exists an algorithm for the problem of maximizing a non-negative monotone submodular function $f : 2^\NN \to \mathbb{R}^+$ over a matroid $\MM = (\NN, \II)$ of rank $k$ with the following properties.
%Given parameters $\ee > 0$ and $\lambda \in [1, k]$:
\begin{enumerate}
\item  $f(S) \geq \left(1 - \nicefrac[]{1}{e} - \ee\right) \cdot \max \left\{ f(T):T\in\II\right\}$.
\item The algorithm performs $O\left(k\lambda + \frac{kn}{\lambda \ee^5} \ln^2\left(\frac{n}{\ee}\right)\right)$ value oracle queries.
\item The algorithm performs $O\left(\frac{k^2}{\ee} + \frac{\lambda n}{\ee^{2}} \ln\left(\frac{n}{\ee}\right)\right)$ independence oracle queries.
\end{enumerate}
\end{theorem}

%It is important to note that the above theorem presents a surprising tradeoff between two seemingly unrelated quantities: the number of value oracle queries and the number of matroid independence queries the algorithm performs.
An example, in which the tradeoff between the number of the two distinct oracle queries is perhaps most insightful, is when $k=\Theta \left( \sqrt{n}\right)$.
%This tradeoff is perhaps most striking in the case $k=\Theta (\sqrt{n})$.
In this case, if one chooses $\lambda=1$ our algorithm performs $\tilde{O} _{\varepsilon}\left( n^{\nicefrac[]{3}{2}}\right)$ value queries, but only $\tilde{O} _{\varepsilon}\left( n\right)$ independence oracle queries.\footnote{Here $\tilde{O} _{\varepsilon}$ hides polylogarithmic factors in $n$ and polynomial factors in $\varepsilon ^{-1}$.}
However, if one chooses $\lambda=k$ the algorithm performs only $\tilde{O} _{\varepsilon}\left( n\right)$ value queries while the number of independence oracle queries grows to $\tilde{O} _{\varepsilon}\left( n^{\nicefrac[]{3}{2}}\right)$. This allows flexibility when the two types of queries have different time complexities in the application at hand. Note that in this case when $k=\Theta \left(\sqrt{n}\right)$ the algorithm of \cite{BV14} corresponds to choosing $\lambda =1$ in our algorithm, as it performs $\tilde{O} _{\varepsilon}\left( n^{\nicefrac[]{3}{2}}\right)$ value queries and $\tilde{O} _{\varepsilon}\left( n\right)$ independence oracle queries.

It is important to note that our algorithm not only establishes a surprising tradeoff between the two different types of queries, but can also provide a significant speedup to the running time when compared to the state of the art algorithm of \cite{BV14}.
Consider, for simplicity, the case where both types of oracles have the same running time.
In this case, obtaining a fast algorithm requires reducing the {\em total} number of oracle queries regardless of their types.
While the algorithm of \cite{BV14} requires $\tilde{O} _{\varepsilon}\left( k^2+nk\right)$ queries, our algorithm requires only $\tilde{O}_{\varepsilon} ( k^2+\sqrt{k}n)$ queries if one sets $\lambda = \sqrt{k}$.
In the above example, where $k=\Theta \left( \sqrt{n}\right)$, it reduces the number of oracle queries from $\tilde{O} _{\varepsilon}\left( n^{\nicefrac[]{3}{2}}\right)$ to $\tilde{O} _{\varepsilon}\left( n^{\nicefrac[]{5}{4}}\right)$.
In fact, if $k=O(n^c)$ for some absolute constant $c \leq 2/3$, $\lambda$ can be chosen such that the our algorithm uses $\tilde{O} _{\varepsilon} \left(n^{1+\nicefrac[]{c}{2}}\right)$ oracle queries, whereas \cite{BV14}'s algorithm needs $\tilde{O}_{\varepsilon} \left( n^{1+c}\right)$ oracle queries (for $2/3 < c < 1$ we still get an improvement, but a smaller one).

\paragraph{Additional Results.}
We consider two \inArxiv{well-studied}\inConference{interesting} special cases of the general problem considered above. The first is the case of a cardinality constraint and the second is the case of a partition matroid independence constraint.
Building upon the ideas developed in the context of the main result, we present even faster algorithms for the above two special cases.
The reader should note that in these two cases the implementation of the independence oracle is trivial, and hence, we focus only on minimizing the number of value oracle queries. The following three theorems summarize these results.

%The ideas we develop allow us to design faster algorithms for the commonly used simpler case of submudular maximization over a partition matroid. For this simpler cases, there is no longer an independence oracle, thus, we are only interested in the number of value oracle queries used by the algorithm (and guaranteeing that the time complexity is bounded by the same expression).

\begin{theorem} \label{th:partition_matroid}
There exists an algorithm that given a non-negative monotone submodular function $f : 2^\NN \to \mathbb{R}^+$, a generalized partition matroid $\MM = (\NN, \II)$ of rank $k$, and a parameter $\ee > 0$, finds a solution $S\in\II$ where:
$f(S) \geq \left( 1-\nicefrac[]{1}{e}-\ee\right)\cdot\max \left\{ f(T):T\in\II\right\}$ and the algorithm performs $ O\left( k\sqrt{\frac{n}{\ee^5}}\ln{\left( \frac{n}{\ee}\right)}+\frac{n}{\ee^5}\ln^2\left( \frac{n}{\ee}\right)\right)$ value oracle queries.

%achieves an approximation ratio of $(1 - e^{-1} - \ee)$ for the problem $\max \{f(S) : S \in \II\}$ using $O(k\sqrt{n\ee^{-5}}\ln(\frac{n}{\ee}) + n\ee^{-5}\ln^2(\frac{n}{\ee}))$ value oracle queries.
\end{theorem}

%Finally, we also consider the even simpler case of a cardinality constraint. For monotone objective functions, we give a formal proof of the following folklore result.

\begin{theorem} \label{th:monotone_cardinality}
There exists an algorithm that given a non-negative monotone submodular function $f : 2^\NN \to \mathbb{R}^+$, and parameters $k\geq 1$ and $\ee >0$,
finds a solution $ S\subseteq \NN$ of size $|S|\leq k$ where: $f(S) \geq \left( 1-\nicefrac[]{1}{e}-\ee\right)\cdot\max \left\{ f(T) : T \subseteq \NN, |T|\leq k\right\}$ and the algorithm performs $ O\left(n \ln \left(\frac{1}{\ee}\right)\right)$ value oracle queries.
%For every constant $\ee > 0$, there exists a $(1 - e^{-1} - \ee)$-approximation algorithm for the problem $\max \{f(S) : |S| \leq k\}$, where $f$ is a non-negative monotone submodular function, using $O(n \ln \ee^{-1})$ value oracle queries.
\end{theorem}

\begin{theorem} \label{th:non_monotone_cardinality}
There exists an algorithm that given a general non-negative submodular function $f : 2^\NN \to \mathbb{R}^+$, and parameters $k\geq 1$ and $\ee >0$,
finds a solution $ S\subseteq \NN$ of size $|S|\leq k$ where: $f(S) \geq \left( \nicefrac[]{1}{e}-\ee\right)\cdot\max \left\{ f(T) : T \subseteq \NN, |T|\leq k\right\}$ and the number of value oracle queries the algorithm performs is $ \min\left\{O\left(\frac{n}{\ee^2}\ln \left(\frac{1}{\ee}\right)\right), O\left(k\sqrt{\frac{n}{\ee} \ln \left(\frac{k}{\ee}\right)} + \frac{n}{\ee} \ln \left(\frac{k}{\ee}\right)\right)\right\}$.
%For every constant $\ee > 0$, there exists a $(e^{-1} - \ee)$-approximation algorithm for the problem $\max \{f(S) : |S| \leq k\}$, where $f$ is a non-negative submodular function, using, in expectation, $\min\{O(n \ee^{-2} \ln \ee^{-1}), O(k\sqrt{n\ee^{-1} \ln (k/\ee)} + n\ee^{-1} \ln (k/\ee))\}$ value oracle queries.
\end{theorem}

The best previously known result for partition matroids is identical to the one that was known for general matroids, \ie, it uses $O\left( \frac{nk}{\varepsilon ^4}\log ^{2}{\left( \frac{n}{\varepsilon}\right)}\right)$ value oracle queries. We note that Theorem \ref{th:monotone_cardinality} is a folklore result that improves over the best previously formally published result of~\cite{BV14}, who described an algorithm using $O(\frac{n}{\ee} \ln \frac{n}{\ee})$ value oracle queries. %For the same constraint with a non-monotone submodular function we prove the following theorem.

\subsection{Additional Related Work}
The literature on submodular maximization is rich and has a long history. We mention here only a few of the most relevant works.
The classical result of Nemhauser et al.~\cite{NWF78} states that the simple discrete greedy algorithm provides an approximation of $\left( 1-\nicefrac[]{1}{e}\right)$ for maximizing a monotone submodular function subject to a cardinality constraint.
This result is known to be tight by the work of Nemhauser et al.~\cite{NW78}.
Feige~\cite{F98} proved the latter holds even when the objective function is restricted to being a coverage function.
Calinescu et al.~\cite{CCPV11} presented the continuous greedy algorithm, which enabled one to achieve the same tight $\left( 1-\nicefrac[]{1}{e}\right)$ guarantee for the more general matroid constraint.

However, when one considers submodular objectives which are not monotone, less is known.
An approximation of $0.309$ was given by Vondr\'{a}k~\cite{V13} for the general matroid independence constraint, which was later improved
to $0.325$ by Oveis Gharan and Vondr\'{a}k~\cite{GV11} using a simulated annealing technique.
Extending the continuous greedy algorithm of~\cite{CCPV11} to general non-negative submodular objectives, Feldman et al.~\cite{FNS11} obtained an improved approximation of $\nicefrac[]{1}{e}-o(1)$ for the same problem.

When considering the special case of a cardinality constraint and a submodular objective which is not necessarily monotone, Buchbinder et al.~\cite{BFNS14} presented a $\nicefrac[]{1}{e}$-approximation algorithm, called ``random greedy'' whose running time is as fast as the discrete greedy algorithm of Nemhauser et al.~\cite{NWF78}.
Furthermore, \cite{BFNS14} also described a slower polynomial time $\left(\nicefrac[]{1}{e} + 0.004\right)$-approximation algorithm, demonstrating that $\nicefrac[]{1}{e}$ is not the right approximation ratio for the problem. On the hardness side, it is known that no polynomial time algorithm can have an approximation ratio better than $0.491$~\cite{GV11}.

\paragraph{Paper Organization.}
Section~\ref{sec:preliminaries} gives general preliminaries. Section~\ref{sec:matroid} describes our results for general and partition matroids (Theorems~\ref{th:general_matroid} and~\ref{th:partition_matroid}). Finally, Sections~\ref{sec:random_sampling} and~\ref{sec:lazy_greedy} prove Theorem~\ref{th:non_monotone_cardinality}. The proof of the folklore result given by Theorem~\ref{th:monotone_cardinality} can be found in Appendix~\ref{app:monotone_cardinality}.
%\inshort{
%In the next sections we prove our main result (Theorem~\ref{th:general_matroid}). The proofs of the other results are omitted from this extended abstract. Section~\ref{sec:preliminaries} gives general preliminaries. Section~\ref{sec:intuition} describes out techniques and the intuition behind our algorithm. Finally, Sections~\ref{sec:main} and~\ref{ssc:random_greedy} give our algorithm and analyze it.
%}

\section{Preliminaries} \label{sec:preliminaries}

Given a non-negative submodular function $f : 2^\NN \to \mathbb{R}^+$, a set $S \subseteq \NN$ and an element $u \in \NN$, we denote by $f(u \mid S) = f(S \cup \{u\}) - f(S)$ the marginal contribution of $u$ to $S$. The following similar lemmata of~\cite{FMV11} and~\cite{BFNS14} are used in many of our proofs.
%\inshort{Finally, we also need the following lemma of~\cite{FMV11}.}

\begin{lemma}[Lemma~2.2 of~\cite{FMV11}] \label{le:distribution_reverse}
Let $f : 2^\NN \to \mathbb{R}$ be submodular. Denote by $A(p)$ a random subset of $A$ where each element appears with probability $p$ (not necessarily independently). Then, $\mathbb{E}[f(A(p))] \geq (1 - p) f(\varnothing) + p \cdot f(A)$.
\end{lemma}

\begin{lemma}[Lemma~2.2 of~\cite{BFNS14}] \label{le:distribution}
Let $f : 2^\NN \to \mathbb{R}^+$ be non-negative and submodular. Denote by $A(p)$ a random subset of $A$ where each element appears with probability at most $p$ (not necessarily independently). Then, $E[f(A(p))] \geq (1 - p) f(\varnothing)$.
\end{lemma}
\section{General Matroid Constraint} \label{sec:matroid}

In this section we describe algorithms for the problem $\max \{f(S) : S \in \II\}$, where $f : 2^\NN \to \mathbb{R}^+$ is a non-negative monotone submodular function and $M = (\NN, \II)$ is a matroid. Throughout the section we use $k$ to denote the rank of $M$ and assume $f(u) \leq f(OPT)$ for every $u \in \NN$. The last assumption can be justified by observing that every element having $f(u) > f(OPT)$ must be a self-loop, and thus, all such elements can be removed from $\MM$ in linear time.

\subsection{Problem Specific Preliminaries}

Given a non-negative submodular function $f : 2^\NN \to \mathbb{R}^+$, its \emph{multilinear} extension is a function $F : [0, 1]^\NN \to \mathbb{R}^+$ defined by $F(x) = \mathbb{E}[f(\RSet(x))]$, where $\RSet(x)$ is a random set containing every element $u \in \NN$ with probability $x_u$, independently. We denote by $\partial_u F(x)$ the derivative of $F$ at point $x$ with respect to the coordinate corresponding to $u$. The multilinear extension has been extensively used for maximizing submodular functions subject to a matroid constrained, starting with the continuous greedy algorithm of Calinescu et al.~\cite{CCPV11}. All algorithms based on the multilinear extension approximate its value at various points using sampling. Unfortunately, this sampling is often responsible for the quite poor time complexity of these algorithms. For this reason, we use a more cautious approach to sampling in this work.

Let $b : \NN \rightarrow \mathbb{R}^+$ be a non-negative function such that $\sum_{u \in S} b(u) \leq f(OPT)$ for every independent set $S \in \II$. Let $m_b$ be a number of samples which is sufficent to approximate $\partial_u F(x)$ up to a multiplicative error of $\delta$ and an additive error of $\delta \cdot b(u)$ with high probability\footnote{By ``high probability'' we mean that the complementary event occurs with a polynomially small probability in $n$.}, for every given choice of $u \in \NN$ and $x \in [0, 1]^\NN$. Badanidiyuru and Vondr\'{a}k~\cite{BV14} describe a version of the continuous greedy algorithm which, together with swap rounding~\cite{CVZ10}, provides an approximation ratio of $(1 - e^{-1} - \delta)$ using $O(m_b n\delta^{-2} \ln(\frac{n}{\delta}))$ value oracle queries and $O(n\delta^{-2} \ln(\frac{n}{\delta}) + k^2\delta^{-1})$ independence oracle queries. Badanidiyuru and Vondr\'{a}k~\cite{BV14} assume $b(u) = f(OPT) / k$ for every $u \in \NN$,\footnote{In fact, \cite{BV14} proves explicitly only this case, but the proof can be easily extended to every function $b$ obeying the condition defined above.} and bound $m_b$ using the following lemma.

\begin{lemma}[Lemma~2.3 of \cite{BV14}]
Let $X_1, X_2, \ldots, X_m$ be independent random variables such that for each $1 \leq i \leq m$, $X_i \in [0, 1]$. Let $X = \frac{1}{m} \cdot \sum_{i = 1}^m X_i$ and $\mu = E[X]$. Then
\begin{align*}
	&
	Pr[X > (1 + \alpha)\mu + \beta] \leq e^{-m\alpha\beta/3}
	\enspace,\\
	&
	Pr[X < (1 - \alpha)\mu - \beta] \leq e^{-m\alpha\beta/2}
	\enspace.
\end{align*}
\end{lemma}

The assumption $f(u) \leq f(OPT)$ for every $u \in \NN$ allowed \cite{BV14} to prove, using the above lemma, that $m_b$ can be set to $k \ln n / \delta^2$. Assume now $\sum_{u \in S} f(u) \leq c \cdot f(OPT)$ for every independent set $S$ and some value $c$, and let us define $b(u) = f(u) / c$. Clearly, $b$ obeys the required condition. Moreover, the above lemma can now be used to show that $m_b$ can be set to $c \ln n / \delta^2$. Thus, we get the following corollary.

\begin{corollary} \label{co:continuous_greedy_guarantee}
If $\max_{S \in \II} \sum_{u \in S} f(u) \leq c \cdot f(OPT)$ for some value $c$, then for every $\delta > 0$ there exists a $(1 - e^{-1} - \delta)$-approximation algorithm for $\max \{f(S) \mid S \in \II\}$ using $O(cn\delta^{-4} \ln^2(\frac{n}{\delta}))$ value oracle queries and $O(n\delta^{-2} \ln(\frac{n}{\delta}) + k^2\delta^{-1})$ independence oracle queries.
\end{corollary}

Some of the algorithms we describe need access to a quick constant approximation of $f(OPT)$. The following lemma provides such an approximation.
\begin{lemma} \label{le:crude_approximation}
There exists a $(1/3)$-approximation algorithm for $\max \{f(S) \mid S \in \II\}$ using $O(n \ln k)$ value and independence oracle queries.
\end{lemma}
The algorithm described by the above lemma is strongly based on the thresholding algorithm of~\cite{BV14}, and thus, we defer the proof of the lemma to Appendix~\ref{app:crude_approximation}. %\inshort{omit its proof from this extended abstract}.

\subsection{Intuition and Techniques} \label{sec:intuition}

Corollary~\ref{co:continuous_greedy_guarantee} tells us that there exists a fast algorithm for the problem $\{f(S) \mid S \in \II\}$ when $\max_{S \in \II} \sum_{u \in S} f(u)$ is not much larger than $f(OPT)$. Thus, we need to show how to deal with the case of large $\max_{S \in \II} \sum_{u \in S} f(u)$. One interesting candidate algorithm for this case is the residual random greedy algorithm suggests by~\cite{BFNS14}. This algorithm works in $k$ iterations. In each iteration, given that $S$ is the current solution of the algorithm, it finds finds a set $S'$ maximizing $\{f(S') \mid S \cup S' \in \II\}$. Then, it selects a random element $u \in S'$, and adds it to its solution $S$.

Buchbinder et al.~\cite{BFNS14} only managed to show that their residual random greedy is a $1/4$-approximation algorithm. However, it is not difficult to check that this algorithm behaves much better as long as \inConference{the value }$\max_{S \cup S' \in \II} \sum_{u \in S'} f(u \mid S)$ is large. More specifically, the expected increase in the value of $S$ is large compared to the expected decrease in the value of the best independent set containing $S$. This suggests the following natural approach. Execute the residual random greedy as long as $\max_{S \cup S' \in \II} \sum_{u \in S'} f(u \mid S)$ is large. Once $\max_{S \cup S' \in \II} \sum_{u \in S'} f(u \mid S)$ becomes small, apply the measured continuous greedy of~\cite{BV14} to the residual problem.

For technical reasons, we also use the observation that the solution produced by the residual random greedy tends to be very small because the value of $S$ increases fast, in expectation. This observation allows us to ignore (``fail'') cases in which the goal of small $\max_{S \cup S' \in \II} \sum_{u \in S'} f(u \mid S)$ is not obtained quickly enough.

\subsection{Main Algorithm} \label{sec:main}

In this section we explain how to combine our variant of the residual random greedy and the measured continous greedy of~\cite{BV14} into an algorithm for $\max\{f(S) \mid S \in \II\}$ having all the properties guaranteed by Theorem~\ref{th:general_matroid}. The following lemma states the properties of our variant of the residual continous greedy that we need. In Section~\ref{ssc:random_greedy}, we describe this variant (which appears as Algorithm~\ref{alg:RandomLazyGreedy1}) and prove Lemma~\ref{le:random_gredy_properties}. In the following, we use $\MM / S$ to denote the matroid obtained from $\MM$ by contracting a set $S \subseteq \NN$.

\begin{lemma} \label{le:random_gredy_properties}
There exists an algorithm that given a non-negative monotone submodular function $f : 2^\NN \to \mathbb{R}^+$, a matroid $\MM = (\NN, \II)$ of rank $k$ and three parameters $\delta \in (0, 1)$, $B \geq 0$ and an integer $0 \leq I \leq k/2$, has the following properties.
\inArxiv{\begin{compactenum}[(i)]}
\inConference{\begin{compactenum}}
	\item The algorithm uses $O(In \delta^{-1} \ln (k / \delta))$ independence oracle queries and $O(Ik + n \delta^{-1} \ln (k / \delta))$ value oracle queries.
	\item The algorithm declares failure with probability at most $k(BI)^{-1}$.
	\item If the algorithm does not fail, it outputs a set $S$ obeying: \label{item:if_succeed}
	\begin{compactitem}
		\item For every independent set $S'$ of\inConference{ the matroid} $\MM / S$, $\sum_{u \in S'} f(u \mid S) \leq (1 - \delta)^{-2}(3B + \delta) \cdot f(OPT)$.
		\item Let $OPT'$ be an independent set of $\MM / S$ maximizing $f(OPT')$. Then, $\mathbb{E}[f(OPT')] \geq [1 - B^{-1}(2 + k/I)] \cdot f(OPT)$, where the expectation is conditioned on the event that the algorithm does not fail.
	\end{compactitem}
\end{compactenum}
\end{lemma}

Algorithm~\ref{alg:Final} is our final algorithm for $\max\{f(S) \mid S \in \II\}$. The algorithm gets two parameters: $\ee > 0$ and $\lambda \in [1, k]$. The last parameter controls a tradeoff between the number of value and independence oracle queries used by the algorithm. In the rest of this section we show that Algorithm~\ref{alg:Final} has all the properties required by Theorem~\ref{th:general_matroid}.

\begin{algorithm*}[th!]
\caption{Combined Algorithm($f, \MM, \ee, \lambda$)} \label{alg:Final}
Call the algorithm guaranteed by Lemma~\ref{le:random_gredy_properties} with the parameters $\delta = 1/2$, $B = 20k\lambda^{-1}\ee^{-1}$ and $I = \lceil \lambda / 3 \rceil)$. Let $S$ denote the output set.\\
\If{the algorithm of Lemma~\ref{le:random_gredy_properties} did not declare failure}
{
	Call the continuous greedy algorithm guaranteed by Corollary~\ref{co:continuous_greedy_guarantee} on the matroid $\MM / S$ and the objective $f(\cdot \mid S)$ with the parameters $c = 240k\lambda^{-1}\ee^{-1} + 2$ and $\delta = \ee/4$. Let $S'$ denote the output set.\\
	\Return{$S \cup S'$}. \\
}
\Else
{
	\Return{$\varnothing$}.\\
}
\end{algorithm*}

\noindent \textbf{Remark}: For $k = 1$, the problem $\max\{f(S) \mid S \in \II\}$ can be solved optimally using $O(n)$ oracle queries. Thus, we assume throughout this section $k \geq 2$. Notice that this assumption implies $I \leq k/2$, as required by Lemma~\ref{le:random_gredy_properties}. Additionally, Theorem~\ref{th:general_matroid} is void for $\ee \geq 1 - e^{-1}$, thus, we also assume $\ee \in (0, 1 - e^{-1})$.

We begin the analysis of Algorithm~\ref{alg:Final} by bounding the number of oracle queries it uses.

\begin{observation}
Algorithm~\ref{alg:Final} uses at most $O(k\lambda + kn\lambda^{-1}\ee^{-5} \ln^2(\frac{n}{\ee}))$ value oracle queries and $O(k^2\ee^{-1} + \lambda n\ee^{-2} \ln(\frac{n}{\ee}))$ independence oracle queries.
\end{observation}
\begin{proof}
The observation follows by adding up the guarantees on the number of oracle queries given by Corollary~\ref{co:continuous_greedy_guarantee} and Lemma~\ref{le:random_gredy_properties}.
\end{proof}

Next, let us lower bound the approximation ratio of Algorithm~\ref{alg:Final}. Let $G$ be the event that the algorithm guaranteed by Lemma~\ref{le:random_gredy_properties} did not declare failure.
\begin{lemma} \label{le:if_not_failing}
Conditioned on $G$, Algorithm~\ref{alg:Final} is $(1 - e^{-1} - \ee/2)$-competitive.
\end{lemma}
\begin{proof}
By Lemma~\ref{le:random_gredy_properties}, conditioned on $G$, there exists a random set $OPT'$ (depending on $S$ only) which is always independent in $\MM / S$ and obeys:
\begin{align*}
	\inConference{&}
	\mathbb{E}[f(OPT') \mid G]
	\geq\inArxiv{{} &}
	\left(1 - \frac{2 + k/I}{B} \right) \cdot f(OPT)\\
	\geq{} &
	\left(1 - \frac{2 + 3k\lambda^{-1}}{20k\lambda^{-1}\ee^{-1}} \right) \cdot f(OPT)
	\geq
	\left(1 - \frac{\ee}{4} \right) \cdot f(OPT)
	\enspace.
\end{align*}

Moreover, Lemma~\ref{le:random_gredy_properties} also guarantees that for every independent set $S'$ of $\MM / S$:
\begin{align*}
	\sum_{u \in S'} f(u \mid& S)
	\leq
	4(3B + 0.5) \cdot f(OPT)\inConference{\\}
	=\inConference{{} &}
	(240k\lambda^{-1}\ee^{-1} + 2) \cdot f(OPT)
	=
	c \cdot f(OPT)
	\enspace.
\end{align*}
Hence, by Corollary~\ref{co:continuous_greedy_guarantee}, given a set $S$, the expected quality of the set produced by the continuous greedy algorithm is at least:
\begin{align*}
	\mathbb{E}[f(S' \mid S)]
	\geq{} &
	\left(1 - \frac{1}{e} - \frac{\ee}{4}\right) \cdot f(OPT' \mid S)\inConference{\\}
	\geq\inConference{{} &}
	\left(1 - \frac{1}{e} - \frac{\ee}{4}\right) \cdot f(OPT') - f(S)
	\enspace.
\end{align*}

Taking now the expectation over all the sets $S$, we get:
\begin{align*}
	\mathbb{E}[f(S \cup S'\inConference{&}) \mid G]
	=\inArxiv{{} &}
	\mathbb{E}[f(S' \mid S) \mid G] + \mathbb{E}[f(S) \mid G]\\
	\geq{} &
	\mathbb{E}\left[\left(1 - \frac{1}{e} - \frac{\ee}{4}\right) \cdot f(OPT') - f(S) ~\middle|~ G\right] \inConference{\\ &}+ \mathbb{E}[f(S) \mid G]\\
	={} &
	\left(1 - \frac{1}{e} - \frac{\ee}{4}\right) \cdot \mathbb{E}[f(OPT') \mid G]\inConference{\\}
	\geq\inConference{{} &}
	\left(1 - \frac{1}{e} - \frac{\ee}{2}\right) \cdot f(OPT)
	\enspace.
	\qedhere
\end{align*}
\end{proof}

\begin{corollary}
\inArxiv{Algorithm~\ref{alg:Final} is a $(1 - e^{-1} - \ee)$-approximation algorithm.}
\inConference{Algorithm~\ref{alg:Final} has an approximation ratio of at least $(1 - e^{-1} - \ee)$.}
\end{corollary}
\begin{proof}
Let $A$ denote the output\inConference{ set} of Algorithm~\ref{alg:Final}. Lemma~\ref{le:if_not_failing} states that:
\[
	\mathbb{E}[f(A) \mid G]
	\geq
	\left(1 - \frac{1}{e} - \frac{\ee}{2}\right) \cdot f(OPT)
	\enspace.
\]

On the other hand, by Lemma~\ref{le:random_gredy_properties}, it is possible to lower bound $\Pr[G]$ by:
\[
	\Pr[G]
	\geq
	1 - \frac{k}{BI}
	\geq
	1 - \frac{k}{[20k\lambda^{-1}\ee^{-1}] \cdot [\lambda/3]}
	\geq
	1 - \frac{3\ee}{20}
	\enspace.
\]

The corollary now follows by observing that:
\begin{align*}
	\mathbb{E}[f(A)]
	\geq{} &
	\Pr[G] \cdot \mathbb{E}[f(A) \mid G] \inConference{\\}
	\geq\inConference{{} &}
	\left(1 - \frac{3\ee}{20} \right) \cdot \left(1 - \frac{1}{e} - \frac{\ee}{2}\right) \cdot f(OPT)\inConference{\\}
	\geq\inConference{{} &}
	\left(1 - \frac{1}{e} - \ee\right) \cdot f(OPT)
	\enspace.
	\qedhere
\end{align*}
\end{proof}

The above corollary completes the proof of Theorem~\ref{th:general_matroid}.

\subsection{A variant of the residual random greedy algorithm} \label{ssc:random_greedy}

In this section we describe an algorithm proving Lemma~\ref{le:random_gredy_properties}. The algorithm we describe is related to the residual random greedy algorithm of~\cite{BFNS14}, with two main modifications. First, the algorithm is sped up using ideas from~\cite{BV14}. Second, the algorithm stops when the total marginal value of all the elements in every independent set becomes small enough. Recall that the last property implies that the measured continuous greedy of~\cite{BV14} can be used efficiently to complete the solution.

The description of our algorithm (give as  Algorithm~\ref{alg:RandomLazyGreedy1}) assumes $\NN$ contains a known set $D$ of $k$ dummy elements having the following properties:
\begin{compactitem}
	\item $f(S) = f(S \setminus D)$ for every set $S \subseteq \NN$.
	\item A set $S \subseteq \NN$ is independent if and only if $S \setminus D$ is independent and $|S| \leq k$.
\end{compactitem}
This assumption can be guaranteed \inConference{to hold }by artificially adding such a set $D$ to the ground set, modifying the value and independence oracles accordingly and removing the elements of $D$ from the \inConference{end }solution\inArxiv{ produced by the algorithm}.

Algorithm~\ref{alg:RandomLazyGreedy1} performs up to $I$ iterations. In each iteration, the algorithm finds an (almost) maximum weight independent set $M$ in the residual matroid (\ie, the matroid resulting from $\MM$ by contracting the current solution $S$), where the weight of an element is defined as its marginal contribution to $S$. If $M$ has a high enough weight, then a random element from it is added to the current solution $S$ and the algorithm continues to the next iteration. Otherwise, the algorithm terminates.

In order to find $M$ using few value oracle queries, the algorithm maintains a global variable $w_u$ for every element $u \in \NN$. The variable $w_u$ is always an upper bound on $f(u \mid S)$.

\SetKwFunction{LinearGreedy}{LinearGreedy}
\SetKwProg{Function}{Function}{}{}
\begin{algorithm*}[ht]
\caption{\textsf{Random Lazy Greedy}$(f, \MM, \delta, B, I)$} \label{alg:RandomLazyGreedy1}
\DontPrintSemicolon
\tcp{Initialization}
Use the algorithm guaranteed by Lemma~\ref{le:crude_approximation} to calculate a value $\opt$ obeying: $f(OPT) \leq \opt \leq 3 \cdot f(OPT)$.\\
Let $S_0 \leftarrow \varnothing$.\\
Let $W \leftarrow \max_{u \in \NN} f(u)$.\\
\lForEach{$u \in \NN$}
{
	Let $w_u \leftarrow W$.
}

\BlankLine

\tcp{Main Loop}
\For{$i$ = $1$ \KwTo $I$}
{
	Let $M_i \leftarrow \LinearGreedy()$.\\
	\If{$(1 - \delta) \cdot \sum_{u \in M_i} w_u \geq B \cdot \opt$}
	{
		Add to $M_i$ enough dummy elements to make $S_{i - 1} \cup M_i$ a base.\\
		Let $u_i$ be a uniformly random element of $M_i$.\\
		Let $S_i \gets S_{i - 1} \cup \{u_i\}$.\\
	}
	\lElse
	{
		\Return{$S_{i - 1}$}.
	}
}
Declare failure. \\

\BlankLine

\Function{$\LinearGreedy()$}
{
	Let $M \leftarrow \varnothing$.\\
	\For{$(w \leftarrow W; w > \delta W / k; w \leftarrow w(1 - \delta))$}
	{
	\ForEach{$u \in \NN$}
		{
			\If{$w_u = w$ and $S_{i - 1} \cup M \cup \{u\} \in \II$ \label{ln:condition_to_M}}
			{
				\lIf{$f(u \mid S_{i - 1}) \leq (1 - \delta) w_u$}
				{
					Update $w_u \leftarrow w_u(1 - \delta)$.
				}
				\lElse
				{
					Add $u$ to $M$.
				}
			}
		}
	}
	\Return{$M$}.\\
}
\end{algorithm*}

We begin the analysis of Algorithm~\ref{alg:RandomLazyGreedy1} by showing it has the complexity required by Lemma~\ref{le:random_gredy_properties}.
%\inshort{The proof of Observation~\ref{ob:random_greedy_complexity} is omitted from this extended abstract.}

\begin{observation} \label{ob:random_greedy_complexity}
Algorithm~\ref{alg:RandomLazyGreedy1} makes $O(In \delta^{-1} \ln (k / \delta))$ independence oracle queries and $O(Ik + n \delta^{-1} \ln (k / \delta))$ value oracle queries.
\end{observation}
\begin{proof}
The algorithm guaranteed by Lemma~\ref{le:crude_approximation} uses only $O(n \ln k)$ oracle queries, which is upper bounded by both guarantees of this observation, and thus, can be ignored. The first part of the observation now follows by multiplying the following values:
\begin{itemize}
	\item Every iteration of the internal loop of \inConference{the algorithm }{\LinearGreedy} uses a single independence oracle query, and this loop repeats $n$ times.
	\item The number of iterations performed by the external loop of {\LinearGreedy} is:
	\[
		\lceil \ln_{1 - \delta} (\delta / k) \rceil
		\leq
		1 - \frac{\ln (k / \delta)}{\ln(1 - \delta)}
		\leq
		1 + \frac{\ln (k / \delta)}{\delta}
		\enspace.
	\]
	\item The main loop of \inArxiv{the algorithm}\inConference{Algorithm~\ref{alg:RandomLazyGreedy1}} repeats \inArxiv{at most}\inConference{up to} $I$ times.
\end{itemize}

The second part of the observation holds since every time a value oracle query is made by Algorithm~\ref{alg:RandomLazyGreedy1}, one of two things must happen: either an element is added to $M$ or $w_u$ is decreased for some element $u \in \NN$. Observe that at most $Ik$ elements can be added to $M$, and the number of times $w_u$ can be decreased for every element $u \in \NN$ is at most:
\[
	\lceil \ln_{1 - \delta} (\delta / k) \rceil
	\leq
	1 + \frac{\ln (k / \delta)}{\delta}
	\enspace.
	\qedhere
\]
\end{proof}

To prove the \inArxiv{other}\inConference{rest of the} properties guaranteed by Lemma~\ref{le:random_gredy_properties}, we first need some notation. Let $i_\ell$ be the (random) largest index for which $S_{i_\ell}$ is defined by the algorithm. For ease of notation, we define for every $0 \leq i \leq I$ the value $r(i) = \min\{i, i_\ell\}$. Notice that $S_{r(i)}$ is always defined, even when $S_i$ is not assigned in a given execution of Algorithm~\ref{alg:RandomLazyGreedy1}. The following lemma lower bounds the expected value of $S_{r(i)}$. 
%\inshort{The main idea of the lemma is that in every iteration $i$ in which an element is added to the solution, the expected increase in the value of the solution is at least $k^{-1} \cdot \sum_{u \in M_i} f(u \mid S_i) \geq (1 - \delta) \cdot \sum_{u \in M_i} w_u \geq B / k$. We omit the proof of the lemma from this extended abstract.}

\begin{lemma} \label{le:lower_bound_using_B}
For every $0 \leq i \leq I$, $\mathbb{E}[f(S_{r(i)})] \geq B \cdot (\mathbb{E}[r(i)] / k) \cdot f(OPT)$.
\end{lemma}
\begin{proof}
We prove the lemma by induction on $i$. For $i = 0$:
\inArxiv{\[}\inConference{$}
	\mathbb{E}[f(S_{r(0)})]
	\geq
	0
	=
	\mathbb{E}[r(0)]
\inArxiv{	\enspace. \]}\inConference{$.}
Next, assume the claim is true for $i - 1 \geq 0$, and let us prove it for $i$. Fix an event $A_i$ specifying the random decisions made by Algorithms~\ref{alg:RandomLazyGreedy1} before iteration $i$. All the probabilities and expectations from this point till we unfix $A_i$ are implicitly conditioned on $A_i$. Notice that once $A_i$ is fixed the sets $S_{i - 1}$ and $M_i$ become deterministic and so is the question whether $i \leq i_\ell$. Based on the last question, we have two cases. If $i \leq i_\ell$, then the definition of Algorithm~\ref{alg:RandomLazyGreedy1} guarantees $\sum_{u \in M_i} f(u \mid S_{i - 1}) \geq (1 - \delta) \cdot \sum_{u \in M_i} w_u \geq B \cdot \opt \geq B \cdot f(OPT)$. Since $u_i$ is a random element from $M_i$,\inArxiv{ we get:}
\begin{align*}
	\mathbb{E}[f(S_{r(i)})\inConference{&}]
	=\inArxiv{{} &}
	f(S_{r(i - 1)}) + \mathbb{E}[f(u_i \mid S_{i - 1})]\inConference{\\}
	\geq\inConference{{} &}
	f(S_{r(i - 1)}) + \frac{\sum_{u \in M_i} f(u \mid S_{i - 1})}{|M_i|}\\
	\geq{} &
	f(S_{r(i - 1)}) + \frac{B \cdot f(OPT)}{k}\inConference{\\}
	=\inConference{{} &}
	f(S_{r(i - 1)}) + \frac{B \cdot [r(i) - r(i - 1)]}{k} \cdot f(OPT)
	\enspace.
\end{align*}
Consider now the case $r(i) < i$. In this case,
\begin{align*}
	\inConference{&}
	\mathbb{E}[f(S_{r(i)})]
	=\inArxiv{{} &}
	f(S_{r(i - 1)})\inConference{\\}
	=\inConference{{} &}
	f(S_{r(i - 1)}) + \frac{B \cdot [r(i) - r(i - 1)]}{k} \cdot f(OPT)
	\enspace.
\end{align*}

In conclusion, \inArxiv{the inequality}\inConference{we have} $\mathbb{E}[f(S_{r(i)})] \geq f(S_{r(i - 1)}) + Bk^{-1} \cdot [r(i) - r(i - 1)] \cdot f(OPT)$ \inArxiv{holds }in both cases. Moreover, since this inequality holds conditioned on every given event $A_i$, it holds also unconditionally. Therefore, unfixing the event $A_i$, we get:
\begin{align*}
	\mathbb{E}[f(S_{r(i)})]
	={} &
	\mathbb{E}[f(S_{r(i - 1)})] \inConference{\\ &} + \frac{B \cdot \mathbb{E}[r(i) - r(i - 1)]}{k} \cdot f(OPT)\\
	\geq{} &
	\frac{B \cdot \mathbb{E}[r(i - 1)]}{k} \cdot f(OPT)  \inConference{\\ &} + \frac{B \cdot \mathbb{E}[r(i) - r(i - 1)]}{k} \cdot f(OPT) \inConference{\\}
	=\inConference{{} &}
	\frac{B \cdot \mathbb{E}[r(i)]}{k} \cdot f(OPT)
	\enspace.
	\qedhere
\end{align*}
\end{proof}

The above lemma implies the following very useful corollary.

\begin{corollary} \label{co:r_bound}
For every $0 \leq i \leq I$, $\mathbb{E}[r(i)] \leq k/B$.
\end{corollary}
\begin{proof}
Notice that $S_{r(i)}$ is always a feasible solution. Thus, $\mathbb{E}[f(S_{r(i)})] \leq f(OPT)$. Using the lower bound on $\mathbb{E}[f(S_{r(i)})]$ given by Lemma~\ref{le:lower_bound_using_B}, we get:
\[
	\frac{\mathbb{E}[r(i)] \cdot B}{k} \cdot f(OPT) \leq f(OPT)
	\Rightarrow
	\mathbb{E}[r(i)] \leq \frac{k}{B}
	\enspace.
	\qedhere
\]
\end{proof}

We are now ready to prove the upper bound on the failure probability of Algorithm~\ref{alg:RandomLazyGreedy1} given by Lemma~\ref{le:random_gredy_properties}. Let $G$ denote the event that Algorithm~\ref{alg:RandomLazyGreedy1} succeeds (\ie, it does not declare failure).
\begin{observation} \label{ob:fail_probability}
$\Pr[\bar{G}] \leq k(BI)^{-1}$.
\end{observation}
\begin{proof}
Notice that Algorithm~\ref{alg:RandomLazyGreedy1} fails exactly when $i_\ell = I$. Hence, by Markov's inequality:
\[
	\Pr[\bar{G}]
	=
	\Pr[i_\ell = I]
	=
	\Pr[r(I) = I]
	\leq
	\frac{\mathbb{E}[r(I)]}{I}
	\leq
	\frac{k}{BI}
	\enspace.
	\qedhere
\]
\end{proof}

Our final objective in this section is to prove Algorithm~\ref{alg:RandomLazyGreedy1} obeys item~\eqref{item:if_succeed} of Lemma~\ref{le:random_gredy_properties}. The following lemma proves the first part of this item.
%\inshort{The main idea of this lemma is that $M_i$ roughly maximizes $\sum_{u \in M_i} f(u \mid M_{i - 1})$ among all feasible sets of $\MM / S_{i - 1}$. We omit the proof from this extended abstract.}

\begin{lemma} \label{le:no_too_good_set}
Conditioned on $G$, every independent set $S' \subseteq \NN \setminus S_{r(I)}$ of the matroid $\MM / S_{r(I)}$ must have: $\sum_{u \in S'} f(u \mid S_{r(I)}) \leq (1 - \delta)^{-2}(3B  + \delta) \cdot f(OPT)$.
\end{lemma}
\begin{proof}
In this proof, given two sets $A, B \subseteq \NN$, we use the notation $f(A : B) = \sum_{u \in A} f(u \mid B)$ to the denote the total marginal contribution to $B$ of $A$'s elements. Since we are conditioned on $G$, Algorithm~\ref{alg:RandomLazyGreedy1} stopped at some iteration $i = r(I) + 1$ after observing $(1 - \delta) \cdot f(M_i : S_{i - 1}) < B \cdot \opt \leq 3B \cdot f(OPT)$. Let $OPT'$ be the set maximizing $f(\cdot : S_{i - 1})$ among the independent subsets of $\MM / S_{i - 1}$. To complete the proof, it is enough to show that $f(M_i : S_{i - 1}) \geq (1 - \delta) \cdot f(OPT' : S_{i - 1}) - \delta \cdot f(OPT)$.

The function {\LinearGreedy} constructs $M_i$ element by element. For every $1 \leq j \leq |M_i|$, let $M_{i, j}$ be the set $M_i$ after $j$ elements are added to it. For consistency, we also define $M_{i, 0} = \varnothing$. Let $OPT'_j$ be a base maximizing $f(\cdot : S_{i - 1})$ among the bases of $\MM / S_{i - 1}$ containing $M_{i, j}$. Clearly, $f(OPT' : S_{i - 1}) = f(OPT'_0 : S_{i - 1})$.

Fix an arbitrary $1 \leq j \leq |M_i|$, and let $v_j = M_{i,j} \setminus M_{i, j - 1}$ be the $j^{th}$ element added to $M_i$. If $v_j \in OPT'_{j - 1}$ then we define $v'_j = v_j$. Otherwise, let $v'_j$ be an arbitrary element of $\NN \setminus M_{i, j}$ belonging to the (single) cycle of $OPT'_{j - 1} \cup \{v_j\}$. Notice that $(OPT'_j \setminus \{v'_j\}) \cup \{v_j\}$ is always a base of $\MM / S_{i - 1}$ containing $M_{i, j}$, and thus, $f(OPT'_j : S_{i - 1}) \geq f(OPT'_{j - 1} \setminus \{v'_j\} \cup \{v_j\} : S_{i - 1})$. Since $M_{i, j - 1} \cup \{v'_j\}$ is independent in $\MM / S_{i - 1}$, at the time point when the algorithm added $v_j$ to $M_i$, the following held:
\begin{align*}
	f(v_j \mid S_{i - 1})
	\geq{} &
	w_{v_j}(1 - \delta)
	=
	w(1 - \delta)\inConference{\\}
	\geq\inConference{{} &}
	w_{v'_j}(1 - \delta)
	\geq
	f(v'_j \mid S_{i - 1}) \cdot (1 - \delta)
	\enspace,
\end{align*}
which implies:
\begin{align*}
	f(OPT'_j : S_{i - 1})&
	\geq
	f(OPT'_{j - 1} \setminus \{v'_j\} \cup \{v_j\} : S_{i - 1})\inConference{\\}
	\geq\inConference{{} &}
	f(OPT'_{j - 1} : S_{i - 1}) - \delta \cdot f(v'_j \mid S_{i - 1})
	\enspace.
\end{align*}

Adding up the above inequality for every $1 \leq j \leq |M_i|$ results in:
\begin{align*}
	f(OPT'_{|M_i|} : S_{i - 1})\inConference{&\\}
	\geq{} &
	f(OPT'_0 : S_{i - 1}) - \delta \cdot \sum_{j = 1}^{|M_i|} f(v'_j \mid S_{i - 1})\inConference{\\}
	\geq\inConference{{} &}
	(1 - \delta) \cdot f(OPT' : S_{i - 1})
	\enspace.
\end{align*}
Finally, every element of $u \in OPT'_{|M_i|} \setminus M_i$, must have $f(u \mid S) \leq \delta W / k \leq \delta \cdot f(OPT) / k$. Thus,
\begin{align*}
	f(OPT'_{|M_i|} \inConference{&} : S_{i - 1}) \inConference{\\}
	\leq{} &
	f(M_i : S_{i - 1}) + |OPT'_{|M_i|} \setminus M_i| \cdot \frac{\delta}{k} \cdot f(OPT)\\
	\leq{} &
	f(M_i : S_{i - 1}) + \delta \cdot f(OPT)
	\enspace.
	\qedhere
\end{align*}
\end{proof}

To prove the second part of item~\eqref{item:if_succeed} of Lemma~\ref{le:random_gredy_properties}, we need the following lemma from~\cite{B69}, which can be found (with a different notation) as Corollary~39.12a in~\cite{S03}.

\begin{lemma} \label{le:perfect_matching_two_bases}
If $A$ and $B$ are two bases of a matroid $\MM = (\NN, \II)$, then there exists a bijection $\phi : A \setminus B \rightarrow B \setminus A$ such for every $u \in B \setminus A$, $A \cup \{u\} \setminus \{\phi(u)\} \in \II$.
\end{lemma}

Using the above lemma we can now define for every $0 \leq i \leq I$, a random variable $OPT_i$ via the following recursive definition.

\begin{compactitem}
	\item $OPT_0$ is an arbitrary base of $\MM$ obtained from $OPT$ by adding enough dummy elements.
	\item For $i > 0$, $OPT_i$ depends on $i_\ell$. If $i > i_\ell$, then $OPT_i$ is an undefined set (its value is never used in the proofs below). Otherwise, if $i \leq i_\ell$, then let $\phi : M_i \rightarrow OPT_{i - 1}$ be a one to one function having the following properties:
	\begin{compactitem}
		\item For every $u \in M_i$, $S_{i - 1} \cup (OPT_{i - 1} \setminus \{\phi(u)\}) \cup \{u\} \in \II$.
		\item For every $u \in OPT_{i - 1} \cap M_i$, $\phi(u) = u$.
	\end{compactitem}
	 Then, $OPT_i = OPT_{i - 1} \setminus \{\phi(u_i)\}$.
\end{compactitem}

The existence of a function $\phi$ having the properties given in the above definition follow from Lemma~\ref{le:perfect_matching_two_bases} since both $OPT_{i - 1} \cup S_{i - 1}$ and $M_i \cup S_{i - 1}$ are bases of $\MM$ (the fact that $OPT_{i - 1} \cup S_{i - 1}$ is a base can be easily verified by induction). When there are multiple possible choices for $\phi$, we assume one is picked based on the contents of the sets $OPT_{i - 1}$ and $S_{i - 1}$ alone.

The following lemma and corollary \inArxiv{show}\inConference{imply} that $OPT_{r(I)}$ is a random set having the properties required by the second part of item~\eqref{item:if_succeed} of Lemma~\ref{le:random_gredy_properties}.

\begin{lemma} \label{le:opt_i_bound}
For every $0 \leq i \leq I$, $\mathbb{E}[f(OPT_{r(i)})] \geq f(OPT) - 2 \cdot \mathbb{E}[r(i)] \cdot f(OPT) / k$.
\end{lemma}
\begin{proof}
We prove the lemma by induction on $i$. For $i = 0$:
\[
	\mathbb{E}[f(OPT_{r(i)})]
	=
	f(OPT_0)
	=
	f(OPT) - \frac{2 \cdot 0}{k} \cdot f(OPT)
	\enspace.
\]
Next, assume the claim holds for $i - 1 \geq 0$, and let us prove it for $i$. Like in the proof of Lemma~\ref{le:lower_bound_using_B}, we fix an event $A_i$ specifying the random decisions made by Algorithms~\ref{alg:RandomLazyGreedy1} before iteration $i$. All the probabilities and expectations from this point till we unfix $A_i$ are implicitly conditioned on $A_i$. Notice that once $A_i$ is fixed the sets $S_{i - 1}$, $M_i$ and $OPT_{i - 1}$ become deterministic and so is the question whether $i \leq i_\ell$. Based on the last question, we have two cases. If $i \leq i_\ell$, then every element of $OPT_{r(i) - 1}$ belongs to $OPT_{r(i)}$ with probability $1 - 1 / |OPT_{r(i) - 1}| = 1 - (k - r(i) + 1)^{-1}$. Thus, by Lemma~\ref{le:distribution_reverse},
\begin{align*}
	\mathbb{E}[f(OP\inConference{&}T_r(i))]
	\geq\inArxiv{{} &}
	\left(1 - \frac{1}{k - r(i) + 1}\right) \cdot f(OPT_{r(i - 1)})\inConference{\\}
	\geq\inConference{&}
	\left(1 - \frac{2}{k}\right) \cdot f(OPT_{r(i - 1)})\\
	\geq{} &
	f(OPT_{r(i - 1)}) - \frac{2}{k} \cdot f(OPT)\inConference{\\}
	=\inConference{&}
	f(OPT_{r(i - 1)}) - \frac{2[r(i) - r(i - 1)]}{k} \cdot f(OPT)
	\enspace,
\end{align*}
where the second inequality holds since $r(i) \leq i \leq I \leq k/2$ and the third inequality holds since $OPT_{r(i - 1)} \subseteq OPT_0$. Consider now the case $i > i_\ell$. In this case:
\begin{align*}
	\mathbb{E}[f(OP\inConference{&}T_{r(i)})]
	=
	f(OPT_{r(i - 1)})\inConference{\\}
	=\inConference{&}
	f(OPT_{r(i - 1)}) - \frac{2[r(i) - r(i - 1)]}{k} \cdot f(OPT)
	\enspace.
\end{align*}

In conclusion, \inArxiv{the inequality }$\mathbb{E}[f(OPT_{r(i)})] \geq f(OPT_{r(i - 1)}) - 2k^{-1}[r(i) - r(i - 1)] \cdot f(OPT)$ holds in both cases. Moreover, since this inequality holds conditioned on every given event $A_i$, it holds also unconditionally. Therefore, unfixing the event $A_i$, we get:
\begin{align*}
	\mathbb{E}[f(\inConference{&}OPT_{r(i)})]\inConference{\\}
	\geq{} &
	\mathbb{E}[f(OPT_{r(i - 1)})] - \frac{2 \cdot \mathbb{E}[r(i) - r(i - 1)]}{k} \cdot f(OPT)\\
	\geq{} &
	f(OPT) - \frac{2 \cdot \mathbb{E}[r(i - 1)]}{k} \cdot f(OPT) \inConference{\\ &}- \frac{2 \cdot \mathbb{E}[r(i) - r(i - 1)]}{k} \cdot f(OPT)\\
	={} &
	f(OPT) - \frac{2 \cdot \mathbb{E}[r(i)]}{k} \cdot f(OPT)
	\enspace.
	\qedhere
\end{align*}
\end{proof}

\begin{corollary}
For every $0 \leq i \leq I$, $\mathbb{E}[f(OPT_{r(i)}) \mid G] \geq [1 - B^{-1}(2 + k/I)] \cdot f(OPT)$.
\end{corollary}
\begin{proof}
By the law of total expectation:
\begin{align*}
	\mathbb{E}[f(OP\inConference{&}T_{r(i)}) \mid G]
	\geq
	\Pr[G] \cdot \mathbb{E}[f(OPT_{r(i)}) \mid G]\inConference{\\}
	=\inConference{{} &}
	\mathbb{E}[f(OPT_{r(i)})] - \Pr[\bar{G}] \cdot \mathbb{E}[f(OPT_{r(i)}) \mid \bar{G}]
	\enspace.
\end{align*}

The term $\mathbb{E}[f(OPT_{r(i)}) \mid \bar{G}]$ can be upper bounded by $f(OPT)$ because $OPT_{r(i)}$ is always a subset of $OPT$ (possibly, plus dummy elements). Combining this observation with Observation~\ref{ob:fail_probability} and Lemma~\ref{le:opt_i_bound} gives:
\begin{align*}
	\mathbb{E}[f(OPT_{r(i)}) \mid G]
	\geq{} &
	\left[f(OPT) - \frac{2 \cdot \mathbb{E}[r(i)]}{k} \cdot f(OPT) \right] \inConference{\\ &}- \frac{k}{BI} \cdot f(OPT)\\
	\geq{} &
	\left[1 - \frac{2 + k/I}{B}\right] \cdot f(OPT)
	\enspace,
\end{align*}
where the last inequality follows from Corollary~\ref{co:r_bound}.
\end{proof}

\subsection{Generalized Partition Matroids}

In this section we prove Theorem~\ref{th:partition_matroid}, which states that for generalized partition matroids it is possible to improve over the result given by Theorem~\ref{th:general_matroid} for general matroids. In this context, there is no longer an independence oracle, thus, we are only interested in the number of value oracle queries used by our algorithm (and guaranteeing that the time complexity is bounded by the same expression). Throughout this section, $\MM$ is a generalized partition matroid and $h \leq k$ is the number of partitions in $\MM$. We denote by $\NN_j \subseteq \NN$, the set of elements in the $j^{th}$ partition of $\MM$ and by $k_j$ the maximum number of elements that can be taken from this partition (\ie, $\sum_{j = 1}^h k_j = k$ and a set $S \subseteq \NN$ is independent if and only if $|S \cap \NN_j| \leq k_j$ for every $1 \leq j \leq h$).

Swap rounding is an algorithm suggested by~\cite{CVZ10} for rounding fractional points in the matroid polytope $\PP(\MM)$ into an (integral) independent set. Badanidiyuru and Jondr\'{a}k~\cite{BV14} observed that swap rounding has a time complexity of $O(bk^2)$ when the fractional point is a convex combination of $b$ independent sets. The following observation states that this bound can be improved for generalized partition matroids.

\begin{observation} \label{ob:pipage_time}
Given a generalized partition matroid and a fractional point $x \in \PP(\MM)$ which is a convex combination of $b$ independent sets, swap rounding can be used to round $x$ using $O(bk)$ time and no value oracle queries.
\end{observation}
\begin{proof}
It is possible to represent independent sets $S \in \II$ in such a way that given an index $1 \leq j \leq h$ one can find an element $u \in S \cap \NN_j$ in $O(1)$ time  (if such an element exists). The pseudo-code of swap rounding given by~\cite{CVZ10} requires only $O(bk)$ time when the sets composing the fractional point are given in a representation having the above property. Moreover, the standard representation of an independent set as a list of items can be converted into a representation having the above property in $O(k)$ time, hence, all $b$ sets can be converted into such a representation in $O(bk)$ time.
\end{proof}

Plugging the above improved time complexity into the result of \cite{BV14} yields the following improved version of Corollary~\ref{co:continuous_greedy_guarantee}.

\begin{corollary} \label{co:continuous_greedy_guarantee_partition}
If $\max_{S \in \II} \sum_{u \in S} f(u) \leq c \cdot f(OPT)$ for some value $c$, then for every $\delta > 0$ there exists a $(1 - e^{-1} - \delta)$-approximation algorithm for maximizing $f$ subject to a generalized partition matroid $\MM$ using $O(cn\delta^{-4} \ln^2(\frac{n}{\delta}))$ value oracle queries and a time complexity bounded by the same expression.
\end{corollary}

To get an improved result for generalized partition matroids, we also need to improve the implementation of the function {\LinearGreedy} of Algorithm~\ref{alg:RandomLazyGreedy1}. We observe that for such matroids one can handle each partition separately in the function {\LinearGreedy}. The separation allows us to remove the element-wise check whether a given element can be added to the current solution, and replace it with a partition-wise check whether the current solution already has the maximum allowed number of partition elements. Additionally, we replace the element specific variables $w_u$ with sets $T_{j, w}$ containing all the elements $u \in \NN_j$ that logically have $w_u = w$. This change allows us to avoid scanning all the elements of $\NN_j$ in order to find the elements $u \in \NN_j$ having $w_u = w$. Finally, for further acceleration, we introduce for each partition a list $\TT_j$ of the non-empty sets $T_{j, w}$.

The improved implementation of {\LinearGreedy} is given as Algorithm~\ref{alg:new_implementation_greedy_linear}. The initialization part should be executed once before the first call to {\LinearGreedy}.

\SetKwFunction{PartitionGreedy}{PartitionGreedy}
\begin{algorithm*}[h!t]
\caption{\textsf{New Implementation of \FuncSty{LinearGreedy}}} \label{alg:new_implementation_greedy_linear}
\DontPrintSemicolon
\tcp{Initialization}
\For{$j$ = $1$ to $k$}
{
	\lFor{$(w \leftarrow W; w > \delta W / k; w \leftarrow w(1 - \delta))$}{$T_{j, w} \gets \varnothing$.}
	\lForEach{$u \in \NN_j$}{Add $u$ to $T_{j, W}$.}
	Let $\TT_j \gets \{T_{j, W}\}$.
}

\BlankLine

\Function{$\LinearGreedy()$}
{
	Let $M \leftarrow \varnothing$.\\
	\lFor{$j$ = $1$ to $h$}{Update $M \gets M \cup \PartitionGreedy(j)$}
	\Return{$M$}.\\
}

\Function{$\PartitionGreedy(j)$}
{
	Let $M_j \leftarrow \varnothing$.\\
	\ForEach{$T_{j, w} \in \TT_j$ in decreasing weight order}
	{
		\ForEach{$u \in T_{j, w}$ \label{line:internal_loop}}
		{
			\lIf{$|M_j| = k_j$}{\Return{$M_j$}.}
			\If{$f(u \mid S_{i - 1}) \leq w(1 - \delta)$}
			{
				Remove $u$ from $T_{j, w}$ and add it to $T_{j, w(1 - \delta)}$ (if such a set exists).\\
				Update $\TT_j$ by removing $T_{j, w}$ if it became empty and adding $T_{j, w(1 - \delta)}$ if it was empty before.\\
			}
			\lElse
			{
				Add $u$ to $M_j$.
			}
		}
	}
	\Return{$M_j$}.\\
}
\end{algorithm*}

\begin{observation} \label{ob:random_greedy_complexity_partition}
Algorithm~\ref{alg:RandomLazyGreedy1} with the new implementation of {\LinearGreedy} given by Algorithm~\ref{alg:new_implementation_greedy_linear} uses $O(Ik + n\delta^{-1} \ln (k / \delta))$ value oracle queries, and has a time complexity bounded by the same expression.
\end{observation}
\begin{proof}
Aside from the $n$ value oracle queries used to calculate $W$ and the $O(n \ln k)$ oracle queries used by the algorithm guaranteed by Lemma~\ref{le:crude_approximation}, every access to $f$ is followed by one of two events: either an element is added to $M$ or the ``logical'' $w_u$ of an element $u$ is reduced. The total number of elements that can be added to $M$ is $Ik$, and the total number of values a single ``logical'' $w_u$ can have is
\[
	\lceil \ln_{1 - \delta} (\delta / k) \rceil
	\leq
	1 - \frac{\ln \left(\frac{k}{\delta} \right)}{\ln(1 - \delta)}
	\leq
	1 + \frac{\ln \left(\frac{k}{\delta} \right)}{\delta}
	\enspace.
\]
This completes the proof of the first part of the observation.

Regarding the time complexity, notice that the main part of Algorithm~\ref{alg:RandomLazyGreedy1} uses $O(n + Ik)$ time (excluding the calls to {\LinearGreedy}) and the initialization step introduced in Algorithm~\ref{alg:new_implementation_greedy_linear} uses $O(n + k \ln_{1 - \delta} (\delta / k)) = O(n\delta^{-1} \ln(k / \delta))$ time. Thus, we only need to bound the time complexity of the new implementation of {\LinearGreedy}.

Each iteration of the loop starting on Line~\ref{line:internal_loop} of Algorithm~\ref{alg:new_implementation_greedy_linear} takes $O(1)$ time (notice that $\TT_j$ can be updated in $O(1)$ time if $\TT_j$ is represented as a double linked list). Moreover, this loop always make at least one iteration, and each iteration (except for maybe one per execution of {\PartitionGreedy}) access $f$. Hence, we can bound the time required for {\LinearGreedy} by $O(h) = O(k)$ plus the number of value oracle queries it uses. The observation now follows since {\LinearGreedy} is called only $I$ times by Algorithm~\ref{alg:RandomLazyGreedy1}.
\end{proof}

\begin{corollary}
Algorithm~\ref{alg:Final} with the new implementation of {\LinearGreedy} given by Algorithm~\ref{alg:new_implementation_greedy_linear} makes at most $O(k\lambda + kn\lambda^{-1}\ee^{-5} \ln^2(\frac{n}{\ee}))$ value oracle queries.
\end{corollary}
\begin{proof}
The observation follows by adding up the guarantees on the number of value oracle queries given by Corollary~\ref{co:continuous_greedy_guarantee_partition} and Observation~\ref{ob:random_greedy_complexity_partition}.
\end{proof}

Theorem~\ref{th:partition_matroid} now follows immediately by the following choice of $\lambda$:\footnote{For $n \leq 2$ this choice might not be in the valid range $[1, k]$, however, for a constant $n$ the problem can be optimally solved using a constant number of value oracle queries.}
\[
	\lambda =
	\begin{cases}
		k & \text{if $k \leq \sqrt{n\ee^{-5}} \ln(\frac{n}{\ee})$} \enspace,\\
		\sqrt{n\ee^{-5}} \ln(\frac{n}{\ee}) & \text{otherwise} \enspace.
	\end{cases}
\]
\section{Random Sampling Algorithms} \label{sec:random_sampling}

In this section we consider an algorithm for the problem $\max \{f(S) : |S| \leq k\}$ based on random sampling. The algorithm has $k$ iterations and two parameters $p \in (0, 1]$ and $1 \leq s \leq \lceil pn \rceil$. In each iteration the algorithm picks a uniformly random sample $M$ of the ground set containing $\lceil pn \rceil$ elements of $\NN$. The elements of $M$ are then assumed to be ordered according to their marginal contribution to the current solution, and a random element out of the top $s$ elements of $M$ is added to the solution. If $s$ is not integral, then each one of the top $\lfloor s \rfloor$ elements of $M$ is added with probability $1/s$ and the $\lceil s \rceil$ element is added with the remaining probability. A formal description of the algorithm is given as Algorithm~\ref{alg:PartialGreedy}. It is important to observe that the sample $M$ can contain elements that already belong to the solution.

\begin{algorithm*}[h!t]
\caption{\textsf{Random Sampling Algorithm}$(f, k, p, s)$} \label{alg:PartialGreedy}
\DontPrintSemicolon
Initialize: $S_0 \leftarrow \varnothing$.\\
\For{$i$ = $1$ \KwTo $k$}
{
    Let $M_i$ be a uniformly random set containing $\lceil pn \rceil$ elements of $\NN$.\\
    Let $d_i$ be a uniformly random value from the range $(0, s]$.\\
		Let $u_i$ be the element of $M_i$ with the $\lceil d_i \rceil$-th largest marginal contribution to $S_{i - 1}$.\\
		\lIf{$f(u_i \mid S_{i - 1}) \geq 0$ \label{ln:avoid_negative}}
		{
			$S_i \leftarrow S_{i - 1} \cup \{u_i\}$.
		}
}
\Return{$S_k$}.
\end{algorithm*}

\begin{observation} \label{ob:genral_complexity}
Algorithm~\ref{alg:PartialGreedy} uses $O(k + nkp)$ value oracle queries.
\end{observation}
\begin{proof}
Algorithm~\ref{alg:PartialGreedy} performs $k$ iterations. Each iteration of Algorithm~\ref{alg:PartialGreedy} requires $O(|M_i|) = O(1 + pn)$ value oracle queries.
\end{proof}

By setting the parameters of Algorithm~\ref{alg:PartialGreedy} to $s = 1$ and $p = \ln \ee^{-1} / k$, we get a folklore algorithm satisfying the properties guaranteed by Theorem~\ref{th:monotone_cardinality} for $\ee \in (e^{-k}, 1 - e^{-1})$ (for $\ee \in (0, e^{-k}]$ the standard greedy algorithm of~\cite{NWF78} fulfills the requirements of the theorem, and for $\ee \geq 1 - e^{-1}$ the theorem is void). For completeness, we prove this folklore result in Appendix~\ref{app:monotone_cardinality}.

In the rest of this section, we use Algorithm~\ref{alg:PartialGreedy} to prove the following theorem. Theorem~\ref{th:non_monotone_cardinality} follows by combining this theorem with the result proved in Section~\ref{sec:lazy_greedy}.

\begin{theorem} \label{th:cardinality_random_nonmonotone}
There exists an algorithm that given a general non-negative submodular function $f : 2^\NN \to \mathbb{R}^+$, and parameters $k\geq 1$ and $\ee >0$,
finds a solution $ S\subseteq \NN$ of size $|S|\leq k$ where: $f(S) \geq \left(\nicefrac[]{1}{e}-\ee\right)\cdot\max \left\{ f(T) : T \subseteq \NN, |T|\leq k\right\}$ and the algorithm performs $O(n \ee^{-2} \ln \ee^{-1})$ value oracle queries.
\end{theorem}

Let $\delta$ be the (single) $\delta > 0$ for which $8 \delta^{-2} \cdot \ln (2\delta^{-1})) = k$. If $\ee \leq \delta$, then the random greedy algorithm of~\cite{BFNS14} can be used to get $(e^{-1})$-approximation using $O(nk) = O(n \delta^{-2} \ln \delta^{-1}) = O(n \ee^{-2} \ln \ee^{-1})$ value oracle queries. On the other hand, if $\ee \geq e^{-1}$, then Theorem~\ref{th:cardinality_random_nonmonotone} is void. Thus, the interesting case, which we assume from now on, is $\ee \in (\delta, e^{-1})$. We set the parameters of Algorithm~\ref{alg:PartialGreedy} as follows: $s = k \lceil pn \rceil / n$ and $p = 8k^{-1}\ee^{-2} \cdot \ln (2\ee^{-1})$. Notice that $p \in (0, 1]$ since $\ee \geq \delta$ and $1 \leq s \leq \lceil pn \rceil$ since $\ee \leq e^{-1}$. Plugging our chosen value of $p$ into Observation~\ref{ob:genral_complexity} yields the time complexity given in Theorem~\ref{th:cardinality_random_nonmonotone}.

Let $A_i$ be an event determining all the random decisions of the algorithm up to iteration $i$ (excluding). In the first part of the proof, we fix an iteration $1 \leq i\leq k$ and an event $A_i$. All the probabilities and expectations in this part of the proof are implicitly conditioned in $A_i$. Notice that $S_{i - 1}$ is a deterministic set when conditioned on $A_i$. We denote by $v_1, v_2, \dotsc, v_k$ the $k$ elements with the maximum marginal contribution to $S_{i - 1}$, sorted in a non-increasing marginal contribution order, and let $X_j$ be an indicator for the event $u_i = v_j$.

\begin{lemma} \label{le:good_common}
$\mathbb{E}[\sum_{j = 1}^k X_j] > 1 - \ee$.
\end{lemma}
\begin{proof}
Let $E$ be the event that $|M_i \cap \{v_1, v_2, \ldots, v_k\}| \geq (1 - \ee/2) s$. Observe that when $E$ occurs there is a range of size at least $(1 - \ee/2) s$ of possible values for $d_i$ that make $\sum_{j = 1}^k X_j$ equal to $1$. Hence,
\[
	\mathbb{E}\left[\sum_{j = 1}^k X_j ~\middle|~ E\right]
	\geq
	\frac{(1 - \ee/2) s}{s}
	=
	1 - \ee/2
	\enspace.
\]

The random variable $|M_i \cap \{v_1, v_2, \ldots, v_k\}|$ has an hypergeometric distribution, and thus, obeys the Chernoff bounds (see, \eg, Theorem~1.17 of~\cite{D11}). Notice that the expectation of this random variable is $s$, hence,
\begin{align*}
    \Pr[\bar{E}]
		={} &
		\Pr[|M_i \cap \{v_1, v_2, \ldots, v_k\}| \leq (1 - \ee/2) s]\inConference{\\}
    \leq\inConference{{} &}
    e^{-\frac{(\ee/2)^2 \cdot s}{2}}
    \leq
    e^{-\frac{\ee^2 kp}{8}}
    =
    \ee / 2
    \enspace.
\end{align*}

Combining the two above inequalities, we get:
\begin{align*}
    \mathbb{E}\left[\sum_{j = 1}^k X_j\right]
    \geq{} &
    \mathbb{E}\left[\sum_{j = 1}^k X_j ~\middle|~ E\right] \cdot (1 - \Pr[\bar{E}])\inConference{\\}
    \geq\inConference{{} &}
    (1 - \ee/2)^2
    >
    1 - \ee
    \enspace.
    \qedhere
\end{align*}
\end{proof}

Given two elements $v_{j_1}$ and $v_{j_2}$ with $j_1 < j_2$, we expect $v_{j_1}$ to have at least as high a probability to be $u_i$ as $v_{j_2}$. This is proved formally by the following lemma.
\begin{lemma} \label{le:non-decreasing}
$\mathbb{E}[X_j]$ is a non-increasing function of $j$.
\end{lemma}
\begin{proof}
Fix an arbitrary pair of indexes $1 \leq j_1 < j_2 \leq k$. We have to prove that $\mathbb{E}[X_{j_1}] \geq \mathbb{E}[X_{j_2}]$. For every set $S \subseteq \NN$, let $\sigma(S)$ be the following set:
\[
	\sigma(S) = \begin{cases}
		S & \text{if \inArxiv{$v_{j_1}, v_{j_2} \in S$ or $v_{j_1}, v_{j_2} \not \in S$}\inConference{$|\{v_{j_1}, v_{j_2}\} \cap S| \in \{0, 2\}$}} \enspace, \\
		S \cup \{v_{j_2}\} \setminus \{v_{j_1}\} & \text{if $v_{j_1} \in S$ and $v_{j_2} \not \in S$} \enspace, \\
		S \cup \{v_{j_1}\} \setminus \{v_{j_2}\} & \text{if $v_{j_2} \in S$ and $v_{j_1} \not \in S$} \enspace.
	\end{cases}
\]

Since $|S| = |\sigma(S)|$, we get $\Pr[M_i = S] = \Pr[M_i = \sigma(S)]$ for every set $S \subseteq \NN$. Moreover, it is not difficult to verify that by the definition of Algorithm~\ref{alg:PartialGreedy}, $\mathbb{E}[X_{j_1} \mid M_i = S] \geq \mathbb{E}[X_{j_2} \mid M_i = \sigma(S)]$. Combining all these observations yields:
\begin{align*}
	&
	\mathbb{E}[X_{j_1}]
	=
	\sum_{S \subseteq \NN} \Pr[M_i = S] \cdot \mathbb{E}[X_{j_1} \mid M_i = S]\inConference{\\}
	\geq\inConference{{} &}
	\sum_{S \subseteq \NN} \Pr[M_i = \sigma(S)] \cdot \mathbb{E}[X_{j_2} \mid M_i = \sigma(S)]
	=
	\mathbb{E}[X_{j_2}]
	\enspace,
\end{align*}
where the last equality uses the observation that $\sigma$ is a bijection.
\end{proof}

Using the two above lemmata, it is now possible to lower bound the expected gain of Algorithm~\ref{alg:PartialGreedy} in iteration $i$.
\begin{lemma} \label{le:improvement_fixed}
\inArxiv{$}\inConference{\[}\mathbb{E}[f(S_i) - f(S_{i - 1})] \geq (1 - \ee) \cdot \frac{f(OPT \cup S_{i - 1}) - f(S_{i - 1})}{k}\inArxiv{$.}\inConference{\enspace.\]}
\end{lemma}
\begin{proof}
Observe that:
\begin{align*}
    \mathbb{E}[f(S_i) - f(\inConference{&}S_{i - 1})]\inConference{\\}
    ={} &
    \mathbb{E}[\max\{f(u_i \mid S_{i - 1}), 0\}]\inConference{\\}
    \geq\inConference{{} &}
    \sum_{j = 1}^k \left[\mathbb{E}[X_j] \cdot \max\{f(v_j \mid S_{i - 1}), 0\}\right]\\
    \geq{} &
    \frac{\sum_{j = 1}^k \mathbb{E}[X_j] \cdot \sum_{j = 1}^k \max\{f(v_j \mid S_{i - 1}), 0\}}{k}
		\enspace.
\end{align*}
where the last inequality holds by Chebyshev's sum inequality since $\max\{f(v_j \mid S_{i - 1}), 0\}$ is non-increasing in $j$ by definition and $\mathbb{E}[X_j]$ is non-increasing in $j$ by Lemma~\ref{le:non-decreasing}. By the definition of the $v_j$'s and the submodularity of the objective:
\begin{align*}
		\sum_{j = 1}^k \max\{f(v_j \mid S_{i - 1}), 0\}
		\geq{} &
    \sum_{u \in OPT} f(u \mid S_{i - 1})\inConference{\\}
    \geq\inConference{{} &}
    f(OPT \cup S_{i - 1}) - f(S_{i - 1})
    \enspace.
\end{align*}
To complete the proof of the lemma recall that, by Lemma~\ref{le:good_common}, $\mathbb{E}[X_j] \geq 1 - \ee$.
\end{proof}

At this point we unfix the event $A_i$. The probabilities and expectations in the rest of this section are no longer implicitly conditioned on $A_i$.

\begin{corollary}
For every $0 \leq i \leq k$, $\mathbb{E}[f(S_i)] \geq (i/k) \cdot [(1 - 1/k)^{i - 1} - \ee] \cdot f(OPT)$.
\end{corollary}
\begin{proof}
First, notice that since Lemma~\ref{le:improvement_fixed} holds for every given event $A_i$, it holds in expectation also unconditionally. More formally, we get for every $1 \leq i \leq k$,
\begin{align*}
	\mathbb{E}[f(S_i) - f(\inConference{&}S_{i - 1})]\inConference{\\}
	\geq{} &
	(1 - \ee) \cdot \frac{\mathbb{E}[f(OPT \cup S_{i - 1})] - \mathbb{E}[f(S_{i - 1})]}{k}
	\enspace.
\end{align*}

Let us lower bound $\mathbb{E}[f(OPT \cup S_{i - 1})]$. Algorithm~\ref{alg:PartialGreedy} adds each element to its solution with probability at most: $(\lceil pn \rceil/n)/s = 1/k$. Hence, each element belongs to $S_{i - 1}$ with probability at most $1 - (1 - 1/k)^{i - 1}$. Let $h(S) = h(S \cup OPT)$. Since $h$ is a non-negative submodular function, we get by Lemma~\ref{le:distribution},
\begin{align*}
    \mathbb{E}[f(OPT \cup S_i)]
    ={} &
    \mathbb{E}[h(S_i)]
    \geq
    (1 - 1/k)^i \cdot h(\varnothing)\inConference{\\}
    =\inConference{{} &}
    (1 - 1/k)^i \cdot f(OPT)
    \enspace.
\end{align*}

Combining the two above inequalities yields,
\begin{align*}
    \mathbb{E}[f(S_i)\inConference{&} - f(S_{i - 1})]\inConference{\\}
    \geq{} &
    (1 - \ee) \cdot \frac{(1 - 1/k)^{i - 1} \cdot f(OPT) - \mathbb{E}[f(S_{i - 1})]}{k}\\
    \geq{} &
    \frac{[(1 - 1/k)^{i - 1} - \ee] \cdot f(OPT) - \mathbb{E}[f(S_{i - 1})]}{k}
    \enspace.
\end{align*}

We are now ready to prove the corollary by induction on $i$.  For $i = 0$, the corollary holds since $f(S_0) \geq 0 = (0/k) \cdot [(1 - 1/k)^{-1} - \ee] \cdot f(OPT)$. Assume the corollary holds for $i - 1 \geq 0$, let us prove it for $i$.
\begin{align*}
    \mathbb{E}[f(S_i)\inConference{&}]
    \geq\inArxiv{{} &}
    \inArxiv
		{
			\mathbb{E}[f(S_{i - 1})] + \frac{\left[\left(1 - \frac{1}{k}\right)^{i - 1} - \ee\right] \cdot f(OPT) - \mathbb{E}[f(S_{i - 1})]}{k}\\
			={} &
		}
    (1 - 1/k) \cdot \mathbb{E}[f(S_{i - 1})] \inConference{\\ &}+ \frac{[(1 - 1/k)^{i - 1} - \ee] \cdot f(OPT)}{k}\\
    \geq{} &
    (1 - 1/k) \cdot \frac{i - 1}{k} \cdot [(1 - 1/k)^{i - 2} - \ee] \cdot f(OPT) \inConference{\\ &}+ \frac{[(1 - 1/k)^{i - 1} - \ee] \cdot f(OPT)}{k}\\
    \geq{} &
    \frac{i}{k} \cdot [(1 - 1/k)^{i - 1} - \ee] \cdot f(OPT)
    \enspace.
    \qedhere
\end{align*}

\end{proof}

Plugging $i = k$ into the above lemma yields:
\begin{align*}
	\mathbb{E}[f(S_k)]
	\geq{} &
	[(1 - 1/k)^{k - 1} - \ee] \cdot f(OPT)\inConference{\\}
	\geq\inConference{{} &}
	(e^{-1} - \ee) \cdot f(OPT)
	\enspace,
\end{align*}
which completes the proof of the approximation ratio guaranteed by Theorem~\ref{th:cardinality_random_nonmonotone}.
\section{Cardinality Constraint via Thresholding} \label{sec:lazy_greedy}

In this section we describe an algorithm for the problem $\max \{f(S) : |S| \leq k\}$ based on a combination of the Random Greedy of~\cite{BFNS14} and the thresholding algorithm of~\cite{BV14}. We prove that this algorithm obeys all the requirements of the following theorem. Theorem~\ref{th:non_monotone_cardinality} follows by combining this theorem with the result proved in Section~\ref{sec:random_sampling}.

\begin{theorem} \label{th:cadinality_lazy}
There exists an algorithm that given a general non-negative submodular function $f : 2^\NN \to \mathbb{R}^+$, and parameters $k\geq 1$ and $\ee >0$,
finds a solution $ S\subseteq \NN$ of size $|S|\leq k$ where: $f(S) \geq \left(\nicefrac[]{1}{e}-\ee\right)\cdot\max \left\{f(T) : T\subseteq \NN, |T|\leq k\right\}$ and the algorithm performs $O(k\sqrt{n\ee^{-1} \ln (k/\ee)} + n\ee^{-1} \ln (k/\ee))$ value oracle queries.
\end{theorem}

\SetKwFunction{FillM}{FillM}
\SetKw{Yield}{yield}

All algorithms considered in this section assume the existence of a set $D \subseteq \NN$ of at least $2k$ dummy elements having a weight of $0$. This assumption can be justified by explicitly adding such a set of elements to the ground set. We may also assume $\ee < e^{-1}$, since the theorem is void otherwise. The first algorithm we consider is Algorithm~\ref{alg:RandomLazyGreedySimple}. This algorithm accepts an error parameter $\delta \in (0, e^{-1})$ which we set later. The function {\FillM} described in Algorithm~\ref{alg:RandomLazyGreedySimple} has a somewhat non-standard semantics. Namely, at the first time it is called it executes from its beginning till reaching the command {\Yield}. At every additional call, {\FillM} starts executing from the place where it stopped on the last call, and continues till reaching the command {\Yield} again.

On a more intuitive level, Algorithm~\ref{alg:RandomLazyGreedySimple} constructs a solution using $k$ iterations. In each iteration the function {\FillM} is used to fill the set $M$ with the $k$ elements having the largest marginal contributions (up to an approximation error). Then, a uniformly random element of $M$ is added to the solution, and elements of $M$ whose marginal contribution decreased significantly following the addition are removed from $M$.

\SetKwFor{Do}{do}{}{}
\SetKwProg{Function}{Function}{}{}
\begin{algorithm*}[h!t]
\caption{\textsf{Random Lazy Greedy Simple}$(f, k, \delta)$} \label{alg:RandomLazyGreedySimple}
\DontPrintSemicolon

\tcp{Initialization}
$M, S_0 \gets \varnothing$.\\
Let $w, W \gets \max_{u \in \NN} f(u)$.\\

\BlankLine

\tcp{Main Loop}
\For{$i$ = $1$ \KwTo $k$}
{
	Call \FillM($M$).\\
	Uniformly pick a random element $u_i$ from $M$.\\
	Let $S_i \gets S_{i - 1} \cup \{u_i\}$.\\
	\ForEach{element $u \in M$}
	{
		\lIf{$f(u \mid S_i) \leq w(1 - \delta)$}{Remove $u$ from $M$.}
	}
}
\Return{$S_k$}.\\

\BlankLine

\Function{\FillM($M$)}
{
	\For{($w = W$; $w > \delta W / k$; $w \leftarrow w(1 - \delta)$)}
	{
		\ForEach{$u \in \NN$}
		{
			\If{$f(u \mid S) > w(1 - \delta)$}
			{
				Add $u$ to $M$.\\
				\lIf{$|M| = k$}{\Yield.}
			}
		}
	}
	\Do{forever}
	{
		Add $k - |M|$ dummy elements of $D \setminus M$ to $M$.\\
		\Yield.\\
	}
}
\end{algorithm*}

Let us begin by analyzing the approximation ratio of Algorithm~\ref{alg:RandomLazyGreedySimple}. We need some notation. Let $M_i$ be the set $M$ at the moment the algorithm picks $u_i$ from it. For two sets $A, B \subseteq \NN$, let $f(A : B) = \sum_{u \in A} f(u \mid B)$. Finally, let $O_i \subseteq \NN$ be the (random) subset of size at most $k$ maximizing $f(O_i : S_{i - 1})$.

\begin{lemma} \label{le:m_weight}
For every $1 \leq i \leq k$, $f(M_i : S_{i - 1}) \geq (1 - \delta) \cdot f(O_i : S_{i - 1}) - \delta \cdot f(OPT)$.
\end{lemma}
\begin{proof}
Let $A_i$ be an event determining all the random choices of the algorithm in the first $i - 1$ iterations. We fix an arbitrary such event $A_i$, and prove the lemma conditioned on this event. Notice that if the lemma holds conditioned an arbitrary event $A_i$, then it also holds unconditionally. Observe also that once we fix $A_i$, the sets $M_i$, $O_i$ and $S_{i - 1}$ all become deterministic.

Let $w_i$ be the value of $w$ at the moment Algorithm~\ref{alg:RandomLazyGreedySimple} chooses $u_i$. There are two cases to consider. Assume first $w_i > \delta W / k$. Since the marginals of elements only decrease as the solution increases, every element $u \in \NN$ with $f(u \mid S_{i - 1}) > w_i$ must be in $M_i$. On the other hand, since elements with low marginals are removed from $M_i$, we also have $f(u \mid S_{i - 1}) > w_i(1 - \delta)$ for every element $u \in M_i$. Let $O'_i = \{u \in O_i \mid f(u \mid S_{i - 1}) > w_i\}$. By the above discussion $O'_i \subseteq M_i$. Thus,
\begin{align*}
	f(O_i : S_{i - 1})
	={} &
	f(O'_i : S_{i - 1}) + f(O_i \setminus O'_i : S_{i - 1})\inConference{\\}
	\leq\inConference{{} &}
	f(O'_i : S_{i - 1}) + w_i \cdot |O_i \setminus O'_i|\\
	\leq{} &
	f(O'_i : S_{i - 1}) + w_i \cdot |M_i \setminus O'_i|\inConference{\\}
	\leq\inConference{{} &}
	f(O'_i : S_{i - 1}) + \frac{f(M_i \setminus O'_i : S_{i - 1})}{1 - \delta}\inConference{\\}
	\leq\inConference{{} &}
	\frac{f(M_i : S_{i - 1})}{1 - \delta}
	\enspace,
\end{align*}
where the second inequality holds since $|M_i| = k \geq |O_i|$ and $O'_i \subseteq O_i \cap M_i$. The last inequality holds since $f(u \mid S_{i - 1}) \geq 0$ for every $u \in O'_i$. This completes the proof of the lemma for the case $w_i > \delta W / k$. Assume, now, $w_i \leq \delta W / k$. In this case every element $u \in \NN$ having $f(u \mid S_{i - 1}) > \delta W / k$ must be in $M_i$. Let $O'_i = \{u \in O_i \mid f(u \mid S_{i - 1}) > \delta W / k\}$. Again $O'_i \subseteq M_i$, and thus,
\begin{align*}
	f(O_i : S_{i - 1})
	={} &
	f(O'_i : S_{i - 1}) + f(O_i \setminus O'_i : S_{i - 1})\\
	\leq{} &
	f(O'_i : S_{i - 1}) + \frac{\delta W}{k} \cdot |O_i \setminus O'_i|\inConference{\\}
	\leq\inConference{{} &}
	f(M_i : S_{i - 1}) + \delta W
	\enspace.
\end{align*}
The lemma now follows \inArxiv{by observing that}\inConference{since} $W \leq f(OPT)$.
\end{proof}

The above lemma can be used to derive a lower bound on the expected improvement in the solution of Algorithm~\ref{alg:RandomLazyGreedySimple} in a given iteration.
\begin{lemma} \label{le:improvement_cardinality_lazy}
For every \inConference{value }$1 \leq i \leq k$, $\mathbb{E}[f(S_i)] \geq \frac{[(1-\delta)(1 - 1/k)^{i - 1} - \delta] \cdot f(OPT) + (k - 1) \cdot \mathbb{E}[f(S_{i - 1})]}{k}$.
\end{lemma}
\begin{proof}
In every given iteration, Algorithms~\ref{alg:RandomLazyGreedySimple} adds every element $u \in \NN$ to its solution with probability at most $1/k$. Hence, for every $u \in \NN$, $\Pr[u \in S_i] \leq 1 - (1 - 1/k)^i$, which implies $\mathbb{E}[f(OPT \cup S_i)] \geq (1 - 1/k)^i \cdot f(OPT)$ by Lemma~\ref{le:distribution}. Hence, by the definition of $O_i$ and the submodularity of $f$:
\begin{align*}
	\mathbb{E}[f(O_i : S_{i - 1})]\inConference{&}
	\geq\inArxiv{{} &}
	\mathbb{E}[f(OPT : S_{i - 1})]\\
	\geq{} &
	\mathbb{E}[f(OPT \cup S_{i - 1}) - f(S_{i - 1})]\inConference{\\}
	\geq\inConference{{} &}
	(1 - 1/k)^{i - 1} \cdot f(OPT) - \mathbb{E}[f(S_{i - 1})]
	\enspace.
\end{align*}

Recall that a uniformly random element of $M_i$ is added to $S_{i - 1}$ to form $S_i$. This observation, together with Lemma~\ref{le:m_weight}, give:
\begin{align*}
	\inConference{&}
	\mathbb{E}[f(S_i) - f(S_{i - 1})]
	=\inArxiv{{} &}
	\frac{\mathbb{E}[f(M_i : S_{i - 1})]}{k}\inConference{\\}
	\geq\inConference{{} &}
	\frac{(1 - \delta) \cdot \mathbb{E}[f(O_i : S_{i - 1})] - \delta \cdot f(OPT)}{k}\\
	\geq{} &
\inArxiv{
	\frac{(1 - \delta) \cdot \left\{(1 - 1/k)^{i - 1} \cdot f(OPT) - \mathbb{E}[f(S_{i - 1})]\right\}  - \delta \cdot f(OPT)}{k}\\
	\geq{} &
}
	\frac{[(1 - \delta)(1 - 1/k)^{i - 1} - \delta] \cdot f(OPT) - \mathbb{E}[f(S_{i - 1})]}{k}
	\enspace.
	\qedhere
\end{align*}
\end{proof}

We are now ready to prove the approximation ratio of Algorithm~\ref{alg:RandomLazyGreedySimple}.

\begin{corollary}
\inArxiv{Algorithm~\ref{alg:RandomLazyGreedySimple} is a $(e^{-1} - 2\delta)$-approximation algorithm.}
\inConference{Algorithm~\ref{alg:RandomLazyGreedySimple} has an approximation ratio of at least $(e^{-1} - 2\delta)$.}
\end{corollary}
\begin{proof}
\inConference{For simplicity, let us denote $\Phi(i) = (1 - \delta)(1 - 1/k)^{i - 1} - \delta$. }
We first prove by induction on $i$ that for $0 \leq i \leq k$: $\mathbb{E}[f(S_i)] \geq \frac{i}{k} \cdot \inArxiv{[(1 - \delta)(1 - 1/k)^{i - 1} - \delta]}\inConference{\Phi(i)} \cdot f(OPT)$. For $i = 0$, the claim is trivial since: $\mathbb{E}[f(S_0)] \geq 0 = \frac{0}{k} \cdot \inArxiv{[(1 - \delta)(1 - 1/k)^{0 - 1} - \delta]}\inConference{\Phi(0)} \cdot f(OPT)$. Next, assume the claim holds for $i - 1 \geq 0$, and let us prove it for $i$. By Lemma~\ref{le:improvement_cardinality_lazy},
\begin{align*}
	\mathbb{E}[f(S_i)]
	\geq{} &
	\frac{\inArxiv{[(1-\delta)(1 - 1/k)^{i - 1} - \delta]}\inConference{\Phi(i)} \cdot f(OPT) + (k - 1) \cdot \mathbb{E}[f(S_{i - 1})]}{k}\\
	\geq{} &
	\inArxiv{
		\frac{[(1-\delta)(1 - 1/k)^{i - 1} - \delta] \cdot f(OPT) + (k - 1) \cdot \frac{i - 1}{k} \cdot [(1 - \delta)(1 - 1/k)^{i - 2} - \delta] \cdot f(OPT)}{k}\\
	}
	\inConference{
		\frac{\Phi(i) \cdot f(OPT)}{k} \\ & + \frac{(k - 1) \cdot \frac{i - 1}{k} \cdot \Phi(i - 1) \cdot f(OPT)}{k}\\
	}
	\geq{} &
	\frac{i}{k} \cdot \inArxiv{[(1 - \delta)(1 - 1/k)^{i - 1} - \delta]}\inConference{\Phi(i)} \cdot f(OPT)
	\enspace.
\end{align*}
For $i = k$, we get:
\begin{align*}
	\inConference{&}
	\mathbb{E}[f(S_k)]
	\geq\inArxiv{{} &}
	\frac{k}{k} \cdot \inArxiv{[(1 - \delta)(1 - 1/k)^{k - 1} - \delta]}\inConference{\Phi(k)} \cdot f(OPT)\\
	\geq{} &
	[(1 - \delta)e^{-1} - \delta] \cdot f(OPT)
	\geq
	[e^{-1} - 2\delta] \cdot f(OPT)
	\enspace.
	\qedhere
\end{align*}
\end{proof}

Choosing $\delta = \ee/2$, the above corollary yields the approximation ratio guaranteed by Theorem~\ref{th:cadinality_lazy}. Let us now analyze the number of value oracle queries made by Algorithm~\ref{alg:RandomLazyGreedySimple}.

\begin{lemma} \label{le:complexity_simple}
Algorithm~\ref{alg:RandomLazyGreedySimple} uses $O(k^2 + n\delta^{-1} \ln (k/\delta)) = O(k^2 + n\ee^{-1} \ln (k/\ee))$ value oracle queries.
\end{lemma}
\begin{proof}
The main body of Algorithm~\ref{alg:RandomLazyGreedySimple} uses $O(k)$ value oracle queries per iteration. Hence, it uses $O(k^2)$ value oracle queries in total. The function {\FillM} uses $O(n)$ value oracle queries for every value $w$ takes. The lemma follows by observing that the total number of values $w$ can take is at most:
\[
	\lceil \ln_{1 - \delta} (\delta / k) \rceil
	\leq
	1 + \frac{\ln (k / \delta)}{-\ln(1 - \delta)}
	\leq
	1 + \frac{\ln (k / \delta)}{\delta}
	\enspace.
	\qedhere
\]
\end{proof}

The $O(k^2)$ term in the guarantee of Lemma~\ref{le:complexity_simple} stems from the fact that Algorithm~\ref{alg:RandomLazyGreedySimple} scans $M$ at the end of each iteration for elements whose marginal contribution became too small. Algorithm~\ref{alg:RandomLazyGreedyImproved} is a variant of Algorithm~\ref{alg:RandomLazyGreedySimple}  which attempts to reduce the number of times $M$ is scanned by first selecting a random element from $M$, and then scanning $M$ only if the selected random element happens to have a small marginal contribution.

\begin{algorithm*}[h!t]
\caption{\textsf{Random Lazy Greedy Improved}$(f, k, \delta)$} \label{alg:RandomLazyGreedyImproved}
\DontPrintSemicolon

\tcp{Initialization}
Let $M, S_0 \gets \varnothing$.\\
Let $w, W \gets \max_{u \in \NN} f(u)$.\\

\BlankLine

\tcp{Main Loop}
Call \FillM($M$).\\
\For{$i$ = $1$ \KwTo $k$}
{
	Uniformly pick a random element $u'_i$ from $M$.\\
	\lIf{$u'_i$ is a dummy element or $f(u'_i \mid S) > (1 - \delta)w$}
	{
		$u_i \leftarrow u'_i$.
	}
	\Else
	{
		\ForEach{$u \in M$}
		{
			\lIf{$u$ is not a dummy element and $f(u \mid S) \leq w(1 - \delta)$}
			{
				Remove $u$ from $M$.
			}
		}
		Call \FillM($M$), and let $\hat{M}$ be the set of elements added to $M$.\\
		Uniformly pick a random element $u_i$ from $\hat{M}$.\\
	}
	Let $S_i \gets S_{i - 1} \cup \{u_i\}$.\\
}
\Return{$S_k$}.\\

\BlankLine

\Function{\FillM($M$)}
{
	\For{($w = W$; $w > \delta W / k$; $w \leftarrow w(1 - \delta)$)}
	{
		\ForEach{$u \in \NN$}
		{
			\If{$f(u \mid S) > w(1 - \delta)$}
			{
				Add $u$ to $M$.\\
				\lIf{$|M| = k$}{\Yield.}
			}
		}
	}
	\Do{forever}
	{
		Add $k - |M|$ dummy elements of $D \setminus M$ to $M$.\\
		\Yield.\\
	}
}
\end{algorithm*}

\begin{observation} \label{ob:at_most_1_over_k}
In every given iteration, Algorithm~\ref{alg:RandomLazyGreedyImproved} adds every element to its solution with probability at most $1/k$.
\end{observation}
\begin{proof}
Fix an iteration $1 \leq i \leq k$, let $M'_i$ and $w'_i$ be the set $M$ at the moment the algorithm selects $u'_i$ and the value of $w$ at that moment. Finally, let $a_i$ be the number of dummy elements and elements with a marginal larger than $(1 - \delta)w_i$ in $M'_i$. The observation clearly holds for every element counted by $a_i$, since such an element cannot get into $\hat{M}$ on this iteration. For an element $u \in \NN$ which is not counted by $a_i$ to enter the solution in this iteration, two events have to happen. First, the algorithm should not select $u_i \gets u'_i$, which happens with probability $(k - a_i) / k$. Second, $u$ has to be selected from $\hat{M}$, which happens with probability at most: $1 / |\hat{M}| = (k - a_i)^{-1}$ (given that the first event happened).
\end{proof}

Using Observation~\ref{ob:at_most_1_over_k}, it is possible to apply to Algorithm~\ref{alg:RandomLazyGreedyImproved} the same analysis used above to lower bound the approximation ratio of Algorithm~\ref{alg:RandomLazyGreedySimple}. For the analysis to work, we need to redefine some notation:
\begin{compactitem}
	\item $M_i$ is a random set determined only by the random decisions of the algorithm before iteration $i$. Given these random decisions, $M_i$ is the set of elements that have a positive probability (in fact $1/k$) to become $u_i$.
	\item $w_i$ is a random value determined only by the random decisions of the algorithm before iteration $i$. Given these random decisions, $w_i$ is the (unique) value that $w$ will take if {\FillM} is called during this iteration.
\end{compactitem}

\inArxiv{The following lemma completes the proof of Theorem~\ref{th:cadinality_lazy}.}
\inConference{Theorem~\ref{th:cadinality_lazy} now follows by the next lemma.}

\begin{lemma} \label{le:complexity_improved}
Algorithm~\ref{alg:RandomLazyGreedyImproved} uses\inArxiv{, in expectation,} $O(k\sqrt{n\delta^{-1} \ln (k/\delta)} + n\delta^{-1} \ln (k/\delta))$ value oracle queries\inConference{, in expectation}.
\end{lemma}
\begin{proof}
The function {\FillM} \inArxiv{uses}\inConference{makes use of} $O(n)$ value oracle queries for every value $w$ takes. Thus, it uses in total $O(n\delta^{-1} \ln (k/\delta))$ queries since the total number of values $w$ can take is at most:
\[
	\lceil \ln_{1 - \delta} (\delta / k) \rceil
	\leq
	1 + \frac{\ln (k / \delta)}{-\ln(1 - \delta)}
	\leq
	1 + \frac{\ln (k / \delta)}{\delta}
	\enspace.
\]

The rest of the proof bounds the expected number of value oracle queries made by the main part of Algorithm~\ref{alg:RandomLazyGreedyImproved}. For every $1 \leq i \leq k$, let $X_i$ be the (random) number of non-dummy elements in $M$ at the beginning of iteration $i$ whose marginal is $w(1 - \delta)$ or less (and thus, will force the algorithm to make value oracle queries if selected as $u'_i$). Clearly, the main part of Algorithm~\ref{alg:RandomLazyGreedyImproved} makes, in expectation, $O(k) \cdot \sum_{i = 1}^k \mathbb{E}[X_i/k] = O(1) \cdot \sum_{i = 1}^k \mathbb{E}[X_i]$ value oracle queries. On the other hand, the expected number of elements added to $M$ in iteration $i$ can be lower bounded by:
\[
	\sum_{j = 0}^k \left[\Pr[X_i = j] \cdot \frac{j}{k} \cdot j \right]
	=
	\frac{\mathbb{E}[X_i^2]}{k}
	\geq
	\frac{(\mathbb{E}[X_i])^2}{k}
	\enspace.
\]

Since the main part of Algorithm~\ref{alg:RandomLazyGreedyImproved} never removes dummy elements from $M$, {\FillM} might add to $M$ up to $2k$ dummy elements and up to $O(n\delta^{-1} \ln (k/\delta))$ other elements. Thus, the following must hold:
{
\inConference{\allowdisplaybreaks}
\begin{align*}
	&
	\sum_{i = 1}^k \frac{(\mathbb{E}[X_i])^2}{k} =	O(n\delta^{-1} \ln (k/\delta)) \inConference{\\}
	\Rightarrow\inConference{{} &}
	\sum_{i = 1}^k \mathbb{E}[X_i] = O(k\sqrt{n\delta^{-1} \ln (k/\delta)})
	\enspace.
	\qedhere
\end{align*}
}
\end{proof}
\section{Conclusion}

We presented fast algorithms for maximizing submodular functions subject to various constraints. Our algorithm for a general matroid constraint has the interesting property that the number of value oracle queries it uses can be reduced at the cost of more independence oracle queries, and vice versa.

As far as we know, such a property did not appear in any previously known algorithm for this problem (or other related problems). Thus, it can be interesting to determine whether this unusual property represents the real nature of the problem's complexity, or is an artifact of our algorithm. 

\apptocmd{\sloppy}{\hbadness 10000\relax}{}{}
\bibliographystyle{plain}
\bibliography{submodular}

\appendix
\section{Proof of Theorem~\ref{th:monotone_cardinality}} \label{app:monotone_cardinality}

In this section we prove the folklore result given by Theorem~\ref{th:monotone_cardinality}. Notice that for $\ee \in (0, e^{-k}]$ the the standard greedy algorithm of~\cite{NWF78} fulfills all the requirements the theorem, and for $\ee \geq 1 - e^{-1}$ the theorem is void. Thus, from this point on we assume $\ee \in (e^{-k}, 1 - e^{-1})$. The algorithm we use to prove Theorem~\ref{th:monotone_cardinality} is Algorithm~\ref{alg:PartialGreedy} with the parameters $s = 1$ and $p = \ln \ee^{-1} / k$. Notice that $p \in (0, 1]$ since $\ee > e^{-k}$. We restate Algorithm~\ref{alg:PartialGreedy} with these parameter values as Algorithm~\ref{alg:PartialGreedyMonotone}. Since the objective function $f$ is assumed to be monotone in Theorem~\ref{th:monotone_cardinality}, the restatement can safely omit the check on Line~\ref{ln:avoid_negative} of Algorithm~\ref{alg:PartialGreedy}.

\begin{algorithm}[h!t]
\caption{\textsf{Random Sampling Algorithm for Monotone Objectives}$(f, k)$} \label{alg:PartialGreedyMonotone}
\DontPrintSemicolon
Initialize: $S_0 \leftarrow \varnothing$.\\
\For{$i$ = $1$ \KwTo $k$}
{
    Let $M_i$ be a uniformly random set containing $\lceil \frac{n \cdot \ln \ee^{-1}}{k} \rceil$ elements of $\NN$.\\
		Let $u_i$ be the element of $M_i$ with the largest marginal contribution to $S_{i - 1}$.\\
		$S_i \leftarrow S_{i - 1} \cup \{u_i\}$.\\
}
\Return{$S_k$}.
\end{algorithm}

First, the following observation follows by plugging our chosen value for $p$ into Observation~\ref{ob:genral_complexity}.

\begin{observation}
Algorithm~\ref{alg:PartialGreedyMonotone} uses $O(n \ln \ee^{-1})$ value oracle queries.
\end{observation}

The following lemma lower bounds the expected improvement in the solution of Algorithm~\ref{alg:PartialGreedyMonotone} in every iteration.

\begin{lemma} \label{le:gain}
For every $1 \leq i \leq k$, $\mathbb{E}[f(u_i \mid S_{i - 1})] \geq (1 - \ee) \cdot \frac{f(OPT) - \mathbb{E}[f(S_{i - 1})]}{k}$.
\end{lemma}
\begin{proof}
Let $A_i$ be an event specifying the random decisions of Algorithm~\ref{alg:PartialGreedyMonotone} up to iteration $i$ (excluding). If the lemma holds conditioned on every given event $A_i$, then it holds also unconditionally. Hence, in the rest of the proof we fix an event $A_i$ and prove the lemma conditioned on this event. All the probabilities and expectations in the proof are implicitly conditioned on $A_i$. Notice that $S_{i - 1}$ is a deterministic set when conditioned on $A_i$.

Let $v_1, v_2, \ldots, v_k$ be the $k$ elements \inArxiv{with}\inConference{having} the largest marginal contributions to $S_{i - 1}$, sorted in a non-increasing marginal contribution order. Additionally, let $X_j$ be an indicator for the event $M_i \cap \{v_1, v_2, \dotsc, v_j\} \neq \varnothing$. Using this notation, it is possible to lower bound $f(u_i \mid S_{i - 1})$ as follows.
\begin{align*}
    f(u_i \mid S_{i - 1}&)
    \geq
    X_k \cdot f(v_k \mid S_{i - 1}) \inConference{\\ &}+ \sum_{j = 1}^{k - 1} X_j \cdot [f(v_j \mid S_{i - 1}) - f(v_{j + 1} \mid S_{i - 1})]
    \enspace.
\end{align*}
On the other hand, for every $1 \leq j \leq k$, we can lower bound $\mathbb{E}[X_j]$ as follows. If $j + \lceil pn \rceil > n$, then $\mathbb{E}[X_j] = 1 \geq 1 - (1 - p)^j$. Otherwise,
\begin{align*}
	\mathbb{E}[X_j]
	={} &
	1 - \frac{\binom{n - j}{\lceil pn \rceil}}{\binom{n}{\lceil pn \rceil}}
	=
	1 - \prod_{r = 0}^{j - 1} \frac{n - \lceil pn \rceil - r}{n - r}\inConference{\\}
	\geq\inConference{{} &}
	1 - \prod_{r = 0}^{j - 1} \frac{n - pn}{n}
	=
	1 - (1 - p)^j
	\enspace.
\end{align*}
Combining the two above observations with the linearity of the expectation, we get:
{
\inConference{\allowdisplaybreaks}
\begin{align*}
    \mathbb{E}[f(\inConference{&}u_i \mid S_{i - 1})]
    \geq\inArxiv{{} &}
    \mathbb{E}[X_k] \cdot f(v_k \mid S_{i - 1}) \inConference{\\ &} + \sum_{j = 1}^{k - 1} \mathbb{E}[X_j] \cdot [f(v_j \mid S_{i - 1}) - f(v_{j + 1} \mid S_{i - 1})]\\
		\geq{} &
		[1 - (1 - p)^k] \cdot f(v_k \mid S_{i - 1}) \inConference{\\ &} + \sum_{j = 1}^{k - 1} [1 - (1 - p)^j] \cdot [f(v_j \mid S_{i - 1}) - f(v_{j + 1} \mid S_{i - 1})]\\
    ={} &
    p \cdot \sum_{j = 1}^k (1 - p)^{j - 1} f(v_j \mid S_{i - 1})
    \enspace,
\end{align*}
}
where the second inequality holds since $f(v_j \mid S_{i - 1}) - f(v_{j + 1} \mid S_{i - 1})$ is always non-negative.

Consider the sum on the rightmost hand side of the above inequality. Every term of this sum is a multiplication of two non-increasing functions of $j$: $(1 - p)^{j - 1}$ and $f(v_j \mid S_{i - 1})$. This allows us to use Chebyshev's sum inequality to bound this sum as follows:
\begin{align*}
    \mathbb{E}[f(u_i \mid S_{i - 1})]\inConference{& \\}
    \geq{} &
    p \cdot \frac{\sum_{j = 1}^k f(v_j \mid S_{i - 1}) \cdot \sum_{j = 1}^k (1 - p)^{j - 1}}{k}\inConference{\\}
    =\inConference{{} &}
    (1 - (1 - p)^k) \cdot \frac{\sum_{j = 1}^k f(v_j \mid S_{i - 1})}{k}\\
    \geq{} &
    (1 - e^{-kp}) \cdot \frac{\sum_{j = 1}^k f(v_j \mid S_{i - 1})}{k}\inConference{\\}
		=\inConference{{} &}
		(1 - \ee) \cdot \frac{\sum_{j = 1}^k f(v_j \mid S_{i - 1})}{k}
    \enspace.
\end{align*}

The lemma now follows by observing that by the definition of the $v_j$'s and the submodularity and monotonicity of $f$,
\begin{align*}
		&
    \sum_{j = 1}^k f(v_j \mid S_{i - 1})
    \geq
    \sum_{u \in OPT} f(u \mid S_{i - 1})\inConference{\\}
    \geq\inConference{{} &}
    f(OPT \cup S_{i - 1}) - f(S_{i - 1})
    \geq
    f(OPT) - f(S_{i - 1})
    \enspace.
		\qedhere
\end{align*}
\end{proof}

\begin{corollary}
For every $0 \leq i \leq k$, $\mathbb{E}[f(S_i)] \geq \left[1 - e^{-\frac{i \cdot (1 - \ee)}{k}}\right] \cdot f(OPT)$.
\end{corollary}
\begin{proof}
Let us denote $\alpha = k^{-1}(1 - \ee)$. Then, by Lemma~\ref{le:gain}, for every $1 \leq i \leq k$,
\begin{align*}
    \mathbb{E}[f(S_i) - f(S_{i - 1})]
    ={} &
    \mathbb{E}[f(u_i \mid S_{i - 1})]\inConference{\\}
    \geq\inConference{{} &}
    \alpha[f(OPT) - \mathbb{E}[f(S_{i - 1})]]
    \enspace.
\end{align*}
Rearranging, we get:
\[
    f(OPT) - \mathbb{E}[f(S_i)]
    \leq
    (1 - \alpha) \cdot [f(OPT) - \mathbb{E}[f(S_{i - 1})]]
    \enspace.
\]
Combining the above inequalities gives:
\begin{align*}
    f(OPT) - \mathbb{E}[f(S_i)]
    \leq{} &
    (1 - \alpha)^i \cdot [f(OPT) - \mathbb{E}[f(S_0)]]\inConference{\\}
    \leq\inConference{{} &}
    (1 - \alpha)^i \cdot f(OPT)
    \enspace.
\end{align*}
Rearranging once more, yields:
\begin{align*}
    \mathbb{E}[f(S_i)]
    \geq{} &
    \left[1 - (1 - \alpha)^i \right] \cdot f(OPT)\inConference{\\}
		\geq\inConference{{} &}
		\left[1 - e^{-i \cdot \alpha} \right] \cdot f(OPT)
    \enspace.
		\qedhere
\end{align*}
\end{proof}

The above corollary implies that $\mathbb{E}[f(S_k)] \geq (1 - e^{\ee-1}) \cdot f(OPT)$. Theorem~\ref{th:monotone_cardinality} follows by combining this inequality with the following lemma.

\begin{lemma} \label{le:inequality}
$1 - e^{\ee-1} \geq 1 - e^{-1} - \ee$.
\end{lemma}
\begin{proof}
Notice that:
\begin{align*}
		1 - e^{\ee-1} \geq 1 - e^{-1} - \ee
		\Leftrightarrow{} &
    e^{\ee-1} \leq e^{-1} + \ee\inConference{\\}
    \Leftrightarrow\inConference{{} &}
    \ee - 1 \leq \ln(e^{-1} + \ee)
    \enspace.
\end{align*}
By the definition of $\ln$:
\begin{align*}
    \ln(e^{-1} + \ee)
    ={} &
    \int_1^{e^{-1} + \ee} \frac{dx}{x}
    =
    \int_1^{e^{-1}} \frac{dx}{x} + \int_{e^{-1}}^{e^{-1} + \ee} \frac{dx}{x}\inConference{\\}
    \geq\inConference{{} &}
    \ln e^{-1} + \ee \cdot \frac{1}{e^{-1} + \ee}
    \geq
    -1 + \ee
    \enspace,
\end{align*}
where the last inequality holds since $\ee < 1 - e^{-1}$.
\end{proof}
\section{Proof of Lemma~\ref{le:crude_approximation}} \label{app:crude_approximation}

In this section we prove Lemma~\ref{le:crude_approximation}, \ie, we describe a $(1/3)$-approximation algorithm for the problem $\max \{f(S) \mid S \in \II\}$ using $O(n \ln k)$ value and independence oracle queries. The algorithm we describe (given as Algorithm~\ref{alg:LazyGreedy}) is a close variant of an algorithm suggested by~\cite{BV14} for the case of a cardinality constraint. We assume in the analysis of the algorithm that $\ee \in (0, 1)$.

\begin{algorithm}
\caption{\textsf{Thresholding Greedy}$(f, \MM, \ee)$} \label{alg:LazyGreedy}
\DontPrintSemicolon
Let $S \leftarrow \varnothing$.\\
Let $W, w \leftarrow \max_{u \in \NN} f(u)$.\\
\For{$(w \leftarrow W; w > \ee W / k; w \leftarrow w(1 - \ee))$}
{
	\ForEach{$u \in \NN$}
	{
		\lIf{$S \cup \{u\} \in \II$ and $f(u \mid S) \geq w$}
		{
			Add $u$ to $S$.
		}
	}
}
\Return{$S$}.\\
\end{algorithm}

\begin{observation}
Algorithm~\ref{alg:LazyGreedy} outputs an independent set and uses $O(n\ee^{-1} \ln (k / \ee))$ value and independence oracle queries.
\end{observation}
\begin{proof}
The first part of the observation holds since Algorithm~\ref{alg:LazyGreedy} does not add an element $u$ to $S$ unless $S \cup \{u\} \in \II$ before the addition. The second part of the observation follows by multiplying three values:
\begin{compactitem}
	\item Each iteration of the internal loop makes $O(1)$ queries to each oracle.
	\item The internal loop repeats $n$ times.
	\item The number of iterations performed by the external loop is:
	\[
		\lceil \ln_{1 - \ee} (\ee / k) \rceil
		\leq
		1 - \frac{\ln (k / \ee)}{\ln(1 - \ee)}
		\leq
		1 + \frac{\ln (k / \ee)}{\ee}
		\enspace.
		\qedhere
	\]
\end{compactitem}
\end{proof}

Next, let us analyze the approximation ratio of Algorithm~\ref{alg:LazyGreedy}. Let $\ell$ be the size of the solution produced by the algorithm, and let $S_i$ be the set $S$ after $i$ elements were added to it. For consistency, we also define $S_0 = \varnothing$. For every $0 \leq i \leq \ell$, let $OPT_i$ be the maximum value independent set containing $S_i$. The following lemma lower bounds the gain of $S_i$ (as a function of $i$) in terms of the loss of $OPT_i$ (again, as a function of $i$).

\begin{lemma} \label{le:improvement_lost_balance}
For every $1 \leq i \leq \ell$, $(1 - \ee) \cdot [f(OPT_{i - 1}) - f(OPT_i)] \leq  f(S_i) - f(S_{i - 1})$.
\end{lemma}
\begin{proof}
Let $u_i = S_i \setminus S_{i - 1}$ be the $i^{th}$ element added by Algorithm~\ref{alg:LazyGreedy}, and let $w_i$ denote the value of $w$ in the iteration when $u_i$ is picked by the algorithm. By the definition of the algorithm, $f(u \mid S_{i - 1}) \geq w_i$. Let $u^*_i = \arg \max_{u \in OPT_{i - 1} \setminus S_{i - 1}} f(u \mid S_{i - 1})$ be an element with the maximal marginal contribution in $OPT_{i - 1} \setminus S_{i - 1}$. By the definition of $OPT_{i - 1}$, $u^*_i$ can be added to $S_{i - 1}$. Since $u^*_i$ was not added before $u_i$, $f(u^*_u \mid S_{i - 1}) \leq w_i / (1 - \ee)$.

If $u_i \in OPT_i$, then $f(OPT_{i - 1}) = f(OPT_i)$ and the lemma clearly holds since $f(S_i) - f(S_{i - 1}) \geq 0$. Otherwise, consider the set $OPT'_i$ obtained by adding $u_i$ to $OPT_{i - 1}$ and removing an element of $OPT_i \setminus S_i$ from the cycle created (if no cycle is created, we remove no element). Clearly $OPT'_i$ is an independent set containing $S_i$. Thus,
\begin{align*}
	&
	f(OPT_{i - 1}) - f(OPT_i)
	\leq
	f(OPT_{i - 1}) - f(OPT'_i)\inConference{\\}
	\leq\inConference{{} &}
	\max_{u \in OPT_{i - 1} \setminus S_{i - 1}} f(u \mid OPT_{i - 1} \setminus \{u\}) \\
	\leq{} &
	\max_{u \in OPT_{i - 1} \setminus S_{i - 1}} f(u \mid S_{i - 1})
	=
	f(u^*_i \mid S_{i - 1})
	\leq
	\frac{w_i}{1 - \ee}\inConference{\\}
	\leq\inConference{{} &}
	\frac{f(u_i \mid S_{i - 1})}{1 - \ee}
	=
	\frac{f(S_i) - f(S_{i - 1})}{1 - \ee}
	\enspace.
	\qedhere
\end{align*}
\end{proof}

To get an interesting result from the last lemma, we need to show that $f(OPT_\ell)$ is not too large, \ie, $f(OPT_i)$ decreases significantly as a function of $i$.

\begin{lemma} \label{le:final_opt_final_solution}
$f(OPT_\ell) \leq f(S_\ell) + \ee \cdot f(OPT)$.
\end{lemma}
\begin{proof}
Every element of $OPT_\ell \setminus S_\ell$ can be added to $S_\ell$ (since $S_\ell$ is a subset of the independent set $OPT_\ell$). Since none of them is added by Algorithm~\ref{alg:LazyGreedy}, we must have $f(u \mid S_\ell) \leq \ee W/k$ for every $u \in OPT_\ell \setminus S_\ell$. Hence,
\begin{align*}
	f(OPT_\ell) - f(S_\ell\inConference{&})
	\leq
	\sum_{u \in OPT_\ell \setminus S_\ell} f(u \mid S_\ell)\inConference{\\}
	\leq\inConference{{} &}
	\sum_{u \in OPT_\ell \setminus S_\ell} \frac{\ee W}{k}
	\leq
	\ee W
	\leq 
	\ee \cdot f(OPT)
	\enspace,
\end{align*}
where the last inequality follows from the assumption that $f(u) \leq f(OPT)$ for every $u \in \NN$.
\end{proof}

Combining the above lemmata imply the following corollary.

\begin{corollary}
\inArxiv{Algorithm~\ref{alg:LazyGreedy} is a $(1/2 - \ee)$-approximation algorithm for $\max\{f(S) \mid S \in \II\}$.}
\inConference{Algorithm~\ref{alg:LazyGreedy} has an approximation ratio of at least $(1/2 - \ee)$ for $\max\{f(S) \mid S \in \II\}$.}
\end{corollary}
\begin{proof}
Lemma~\ref{le:improvement_lost_balance} implies:
\begin{align*}
	\inConference{&}
	(1 - \ee) \cdot [f(OPT_0) - f(OPT_\ell)]\inConference{\\}
	={} &
	(1 - \ee) \cdot \sum_{i = 1}^\ell [f(OPT_{i - 1}) - f(OPT_i)]\\
	\leq{} &
	\sum_{i = 1}^\ell [f(S_i) - f(S_{i - 1})]
	=
	f(S_\ell) - f(S_0)
	\leq
	f(S_\ell)
	\enspace.
\end{align*}

Rearranging and using Lemma~\ref{le:final_opt_final_solution} and the observation $f(OPT_0) = f(OPT)$, we get:
\begin{align*}
	f(OPT)
	\leq\inConference{{} &}
	f(OPT_\ell) + \frac{f(S_\ell)}{1 - \ee}\inConference{\\}
	\leq\inConference{{} &}
	[f(S_\ell) + \ee \cdot f(OPT)] + \frac{f(S_\ell)}{1 - \ee}\inConference{\\}
	=\inConference{{} &}
	\ee \cdot f(OPT) + \frac{2 - \ee}{1 - \ee} \cdot f(S_\ell)
	\enspace.
\end{align*}
Hence,
\begin{align*}
	f(S_\ell)
	\geq{} &
	\frac{(1 - \ee)^2}{2 - \ee} \cdot f(OPT)\inConference{\\}
	\geq\inConference{{} &}
	\frac{1 - 2\ee}{2} \cdot f(OPT)
	=
	(1/2 - \ee) \cdot f(OPT)
	\enspace.
	\qedhere
\end{align*}
\end{proof}

Lemma~\ref{le:crude_approximation} now follows by choosing $\ee = 1/6$.

\end{document}